%% file: main.tex
\documentclass[11pt,letterpaper]{article}
\usepackage[margin=1in,footskip=0.25in]{geometry}
\usepackage{amssymb}
\usepackage{graphicx}
\usepackage[ruled,linesnumbered,vlined]{algorithm2e}
\usepackage{microtype}
\usepackage{subfig}
\usepackage{amsmath}
\usepackage{amsthm}
\usepackage{enumerate}
\usepackage{wrapfig}
\usepackage{multirow}
\usepackage{nicefrac}
\usepackage[colorinlistoftodos,prependcaption,textsize=tiny]{todonotes}
\usepackage[colorlinks,linkcolor=blue,filecolor=blue,citecolor=blue,urlcolor=blue,pdfstartview=FitH,pagebackref]{hyperref}
\usepackage[nameinlink]{cleveref}
\usepackage{authblk}
\usepackage{multicol}
\usepackage{thm-restate}
\usepackage{microtype}

\usepackage{url}

\definecolor{forestgreen}{rgb}{0.13, 0.55, 0.13}

\SetCommentSty{mycommfont}

\newtheorem{theorem}{Theorem}
\newtheorem{lemma}{Lemma}
\numberwithin{lemma}{section}
\newtheorem{claim}[lemma]{Claim}

\newtheorem{observation}[lemma]{Observation}
\newtheorem{definition}[lemma]{Definition}
\newtheorem{remark}[lemma]{Remark}
\newtheorem{corollary}[lemma]{Corollary}

\newtheorem{problem}{Problem}

\graphicspath{{./fig/}}

\newcommand{\poly}{{\rm poly}}
\newcommand{\polylog}{{\rm polylog}}
\def\C{\mathcal{C}}
\def\G{\mathcal{G}}

\def\D{\mathcal{D}}
\def\O{\mathcal{O}}
\def\P{\mathcal{P}}
\def\Q{\mathcal{Q}}
\def\reals{{\mathbb R}}

\def\R{\mathbb{R}}
\def\N{\mathbb{N}}
\def\eps{{\varepsilon}}
\def\diam{\rm diam}
\def\polylog{\rm polylog}
\def\MEB{\rm MEB}
\def\gMEB{\gamma\mbox{-}\MEB}

\def\init{\rm init}
\def\inc{\rm inc}

\def\dist{\rm dist}
\newcommand{\etal}{{et al. \xspace}}
\newcommand{\argmin}{\rm argmin}
\newcommand{\argmax}{\rm argmax}

\newcommand{\frechet}{Fr\'echet}
\newcommand{\dfd}{d_{dF}}

\def\dfn#1{\emph{\textcolor{forestgreen}{\textbf{#1}}}}

\date{}

\begin{document}
	
	\title{Static and Streaming Data Structures for \frechet\ Distance Queries}
	
	\author[1]{Arnold Filtser\thanks{Supported by the Simons Foundation.}}
	\author[2]{Omrit Filtser\thanks{Supported by the Eric and Wendy Schmidt Fund for Strategic Innovation, by the Council for Higher Education of Israel, and by Ben-Gurion University of the Negev.}}
	
	\affil[1]{Columbia University\\ \texttt{arnold273@gmail.com}}
	\affil[2]{Stony Brook University\\ \texttt{omrit.filtser@gmail.com}}
	
	\maketitle
	\thispagestyle{empty}
	
	\begin{abstract}
		Given a curve $P$ with points in $\mathbb{R}^d$ in a streaming fashion, and parameters $\varepsilon>0$ and $k$, we construct a distance oracle that uses $O(\frac{1}{\varepsilon})^{kd}\log\varepsilon^{-1}$ space, and given a query curve $Q$ with $k$ points in $\mathbb{R}^d$, returns in $\tilde{O}(kd)$ time a $1+\varepsilon$ approximation of the discrete Fr\'echet distance between $Q$ and $P$.
		
		In addition, we construct simplifications in the streaming model, oracle for distance queries to a sub-curve (in the static setting), and introduce the zoom-in problem. Our algorithms work in any dimension $d$, and therefore we generalize some useful tools and algorithms for curves under the discrete Fr\'echet distance to work efficiently in high dimensions.
	\end{abstract}
	
	\vfill
	
	\setlength{\columnsep}{30pt} 
	\begin{multicols}{2}
		{\small \setcounter{tocdepth}{2} \tableofcontents}
	\end{multicols}
	
	\newpage
	\pagenumbering{arabic} 
	
	\input{Introduction}

	\input{Preliminaries}

	\input{Paper_Overview}

	\input{Static_Distance_Oracle}

	\input{Streaming_Simplification}

	\input{Streaming_Distance_Oracle}

	\input{Subcurve_Distance_Oracle}

	\input{High_D_Optimization}

	\newpage
	
	\bibliographystyle{alphaurlinit}
	\bibliography{refs}
	\appendix

	\include{Missing_Proofs}
	
\end{document}

%% file: Introduction.tex
	\section{Introduction}
	Measuring the similarity of two curves or trajectories is an important task that arises in various applications. 
	The \frechet\ distance and its variants became very popular in last few decades, and were widely investigated in the literature. Algorithms for various tasks regarding curves under the \frechet\ distance were implemented, and some were successfully applied to real data sets in applications of computational biology~\cite{JXZ08,WZ13}, coastline matching~\cite{MDLH06}, analysis of a football match~\cite{GW10}, and more (see also GIS Cup SIGSPATIAL'17~\cite{WO17}).

	The \frechet\ distance between two curves $P$ and $Q$ is often described by the man-dog analogy, in which a man is walking along $P$, holding a leash connected to its dog who walks along $Q$, and the goal is to minimize the length of the leash that allows them to fully traverse their curves without backtracking.
	In the discrete \frechet\ distance, only distances between vertices are taken into consideration. Eiter and Mannila~\cite{EM94} presented an $O(nm)$-time simple dynamic programming algorithm to compute the discrete \frechet\ distance of two curves $P$ and $Q$ with $n$ and $m$ vertices. A \polylog\ improvement exists (see~\cite{ABKS14}), however, there is a sequence of papers \cite{Bringmann14,BM16,BOS19} showing that under SETH, there are no strongly subquadratic algorithms for both continuous and discrete versions,
	even if the solution may be approximated up to a factor of 3.
	
	In applications where there is a need to compute the distance to a single curve many times, or when the input curve is extremely large  and quadratic running time is infeasible, a natural solution is to construct a data structure that allows fast distance queries.
	In this paper we are mainly interested in the following problem under the discrete \frechet\ distance. Given $P\in\reals^{d\times m}$ (a $d$-dimensional polygonal curve of length $m$), preprocess it into a data structure that given a query curve $Q\in \reals^{d\times k}$, quickly returns a $(1+\eps)$-approximation of $\dfd(P,Q)$, where $\dfd$ is the discrete \frechet\ distance. Such a data structure is called $(1+\eps)$-distance oracle for $P$.

	Recently, Driemel, Psarros, and Schmidt~\cite{DPS19} showed how to construct a $(1+\eps)$-distance oracle under the discrete \frechet\ distance, with query time that does not depend on $m$, the size of the input curve. Their data structure uses $k^k\cdot O(\frac1\eps)^{kd}\cdot\log^k\frac{1}{\eps}$ space and has $O(k^2d+\log\frac1\eps)$ query time.\footnote{Driemel \etal \cite{DPS19} also considered the more general case where the curves are from a metric space with bounded doubling dimension. We present here only their results for Euclidean space.} 
	They also consider the streaming scenario, where the curve is given as a stream and its length is not known in advance. Their streaming algorithm can answer queries at any point in the stream in $O(k^4d\cdot \log^2\frac{m}{\eps})$ time, and it uses $\log^2m\cdot k^k\cdot O(\frac{\log m}{\eps})^{kd}\cdot \log^k(\frac{\log m}{\eps})$ space.
	Their techniques in the streaming case 
	include a merge-and-reduce framework, which leads to the high query time.
	
	In order to achieve a query time that does not depend on $m$ (in the static case), \cite{DPS19} first compute an (approximation of) optimal $k$-simplification of the input curve $P$. An optimal $k$-simplification of a curve $P$ is a curve $\Pi$ of length at most $k$ which minimizes $\dfd(P,\Pi)$ over all other curves of length at most $k$. Note that as the triangle inequality apply for $\dfd$, a trivial $3$-distance oracle is just computing an optimal $k$-simplification $\Pi$ of $P$, and for a query $Q$ returning $\dfd(P,\Pi)+\dfd(\Pi,Q)$ (see \Cref{obs:trivial_oracle}).
	Specifically, \cite{DPS19} present a streaming algorithm that maintains an 8-approximation for an optimal $k$-simplification of the input curve, and uses $O(kd)$ space. 
	Abam et al.~\cite{ABHZ10} show a streaming algorithm that maintains a simplifications under the continuous \frechet\ distance. Their algorithm maintains a $2k$-simplification which is $(4\sqrt{2}+\eps)$-approximation compared to an optimal $k$-simplification, using $O(k\eps^{-0.5} \log^2\frac1\eps)$ space. 
	In the static scenario, Bereg et. al.~\cite{BJWYZ08} show how to compute an optimal $k$-simplification of a curve $P\in\reals^{3\times m}$ in $O(mk\log m \log(m/k))$ time.
	
	For the (continuous) \frechet\ distance, Driemel and Har-Peled~\cite{DH13} presented a $(1+\eps)$-distance oracle for the special case of $k=2$ (queries are segments). Their data structure uses $O(\frac1\eps)^{2d}\cdot\log^2\frac{1}{\eps}$ space, and has $O(d)$ query time. 
	In addition, they show how to use the above data structure in order to construct a distance oracle for segment queries to a sub-curve (again only for queries of length $k=2$). This data structure uses $m\cdot O(\frac1\eps)^{2d}\cdot\log^2\frac{1}{\eps}$ space, and can answer $(1+\eps)$-approximated distance queries to any subcurve of $P$ in $O(\eps^{-2}\log m\log \log m)$ time.
	In~\cite{Filtser18}, the second author showed how to apply their techniques to the discrete \frechet\ distance, and achieve the same space bound with $O(\log m)$ query time.
	For general $k$, Driemel and Har-Peled~\cite{DH13} provided a constant factor distance oracle that uses $O(md\log m)$ space, and can answer distance queries between any subcurve of $P$ and query $Q$ of length $k$ in $O(k^2d\log m\log(k\log m))$ time.
		
	For the special case where the queries are horizontal segments, de Berg et al. \cite{dBMO17} constructed a data structure that uses $O(n^2)$ space, and can answer exact distance queries (under the continuous \frechet\ distance) in $O(\log^2m)$ time.
		
	The best known approximation algorithm for the discrete \frechet\ distance between two curves $P,Q\in\reals^{d\times m}$ is an $f$-approximation that runs in $O(m\log m + m^2/f^2)$ time for constant $d$, presented by Chan and Rahamati \cite{CR18} (improving over \cite{BM16}). The situation is better when considering restricted (realistic) families of curves such as $c$-packed, $\kappa$-bounded, backbone curves, etc. for which there exists small factor approximation algorithms in near liner time (see e.g. \cite{DHW12,AHKWW06,GMMW19}).
	
	Other related problems include the approximate nearest neighbor problem for curves, where the input is a set of curves that needs to be preprocessed in order to answer (approximated) nearest neighbor queries  (see~\cite{Indyk02,DS17,EP18,ACKGS18,DPS19,FFK20}), and range searching for curves, where the input is a set of curves and the query algorithm has to return all the curves that are within some given distance from the query curve  (see~\cite{dBCG13,dBGM17,BB17,BDDM17,DV17,AD18,FFK20}). We refer to \cite{FFK20} for a more detailed survey of these problems.

	\paragraph{Our results.}
	We consider distance oracles under the discrete \frechet\ distance in both the static and streaming scenarios. See \Cref{tbl:results} for a summary of new and previous results.
	
	In the static case, given an input curve $P\in \reals^{d\times m}$, we construct a $(1+\eps)$-distance oracle with $O(\frac{1}{\eps})^{kd}\cdot\log\frac1\eps$ storage space and $\tilde{O}(kd)$ query time (\Cref{thm:DOasymmetric}). Notice that our bounds in both storage space and query time do not depend on $m$, and are significantly smaller than the bounds of \cite{DPS19}.
	Interestingly, for the streaming setting we manage to achieve the exact same bounds as for the static case (\Cref{thm:mainStreamingDO}). Thus providing a quartic improvement (degree $4$) in the query time compared to \cite{DPS19}. 
	
	As in \cite{DPS19}, we use simplifications to get bounds that do not depend on $m$. Therefore, in the static case we present an algorithm that computes in $\tilde{O}(\frac{md}{\eps^{4.5}})$ time a $(1+\eps)$-approximation for an optimal $k$-simplification of a curve $P\in\reals^{d\times m}$ (\Cref{thm:SimplificationHighDim}).
	Note that the algorithm of \cite{BJWYZ08} returns an optimal $k$ simplification, however, it works only for constant dimension $d$, and has quadratic running time for the case $k=\Omega(m)$.
	For the streaming setting, we present a streaming algorithm which uses $O(\eps)^{-\frac{d+1}{2}}\log^{2}\frac1\eps+O(kd\cdot \frac1\eps\log\frac1\eps)$ 
	space, and computes a $(1+\eps)$-approximation for an optimal $k$-simplification of the input curve (\Cref{cor:StremingSimpEpsilon}).
	In addition, we present a streaming algorithm which uses $O(kd\cdot \frac1\eps\log\frac1\eps)$
	space, and computes a $(1.22+\eps)$-approximation for an optimal $k$-simplification of the input curve (\Cref{cor:StremingSimpConstant}).
		
	We also consider the problem of distance queries to a sub-curve, as in \cite{DH13,Filtser18}. Here, given a curve $P\in \reals^{d\times m}$ (in the static setting), we construct a data structure that uses $m\log m\cdot O(\frac{1}{\eps})^{kd}\cdot\log\frac1\eps$ space, and given a query curve $Q\in \reals^{d\times k}$ and two indexes $1\le i\le j\le m$, returns in $\tilde{O}(k^2d)$ time a $(1+\eps)$-approximation of $\dfd(P[i,j],Q)$, where $P[i,j]$ is the sub-curve of $P$ from index $i$ to $j$ (\Cref{thm:DOsubCurve}). Notice that in this problem the space bound must be $\Omega(m)$, as given such a data structure, one can (essentially) recover the curve $P$. 
	
	Related to both the sub-curve distance oracle and simplifications, we present a new problem called the ``zoom-in'' problem. In this problem, given a curve $P\in \reals^{d\times m}$, our goal is to construct a data structure that given two indexes $1\le i< j\le m$, return an (approximation of) optimal $k$-simplification for $P[i,j]$. This problem is motivated by applications that require visualization of a large curve without displaying all its details, and in addition enables ``zoom-in'' operations, where only a specific part of the curve needs to be displayed. For example, if the curve represents the historical prices of a stock, one might wish to examine the rates during a specific period of time. In such cases, a new simplification needs to be calculated. 
	We present a data structure with $O(mkd\log\frac{m}{k})$ space, such that given a pair of indices $1\le i< j\le m$, returns in $O(kd)$ time a $2k$-simplification which is a $(1+\eps)$-approximation compared to an optimal $k$ simplification of $P[i,j]$.	
	
	Finally, our algorithms work and analyzed for any dimension $d$. Unfortunately, many tools and algorithms that were developed for curves under the discrete \frechet\ distance, considered only constant or low dimensions, and have exponential running time in high dimensions (this phenomena usually referred to as ``the curse of dimensionality''). Therefore, we present a simple technique (\Cref{lem:allDistances}) that allows us to achieve efficient approximation algorithms in high dimensions. Specifically, we use it in  \Cref{thm:SimplificationHighDim} to compute an approximation for an optimal simplification in arbitrary dimension $d$, and to remove the exponential factor from the approximation algorithm of \cite{CR18} (see \Cref{thm:approximation}).

	\begin{table}[h]
		\begin{tabular}{ | p{2.6cm} | p{5.6cm} | p{2.6cm} | p{4cm} |}
			\hline
			 & Space & Time & Comments \\ \hline\hline
			
			\multirow{4}{2.5cm}{Static $(1+\eps)$-distance oracle}  
			& $O(kd)$ & $O(k^2d)$ & $(3+\eps)$-approximation, \Cref{obs:trivial_oracle} \\ \cline{2-4}
			& $O(\frac1\eps)^{2d}\cdot\log^2\frac{1}{\eps}$ & $O(d)$ & {$k=2$, continuous, \cite{DH13}} \\ \cline{2-4}
			& $k^k\cdot O(\frac1\eps)^{kd}\cdot \log^k\frac1\eps$ & $O(k^2d+\log\frac1\eps)$ & {\cite{DPS19}} \\ \cline{2-4}
			& $O(\frac{1}{\eps})^{kd}\cdot\log\frac1\eps$ & $\tilde{O}(kd)$ & {\Cref{thm:DOasymmetric}}\\ \hline\hline
			
			\multirow{2}{2.5cm}{Streaming $(1+\eps)$-distance oracle}  
			& $\log^2m\cdot k^k\cdot O(\frac{\log m}{\eps})^{kd} \cdot\log^k(\frac{\log m}{\eps})$ & $O(k^4d\cdot \log^2\frac{m}{\eps})$ & {\cite{DPS19}} \\ \cline{2-4}
			& $O(\frac{1}{\eps})^{kd}\cdot\log\frac1\eps$ & $\tilde{O}(kd)$ &  {\Cref{thm:mainStreamingDO}} \\  \hline\hline
			
			\multirow{3}{2.5cm}{$(1+\eps)$-distance oracle with subcurve queries}  
			& $m\cdot O(\frac1\eps)^{2d}\cdot\log^2\frac{1}{\eps}$ & $O(\frac{\log m\log \log m}{\eps^{2}})$ & {$k=2$, continuous, \cite{DH13}} \\ \cline{2-4}
			& $m\cdot O(\frac1\eps)^{2d}\cdot\log^2\frac{1}{\eps}$ & $O(\log m)$ & {$k=2$, \cite{Filtser18}} \\ \cline{2-4}
			& $m\log m\cdot O(\frac{1}{\eps})^{kd}\cdot\log\frac1\eps$ & $\tilde{O}(k^2d)$~~\footnotemark & {\Cref{thm:DOsubCurve}}\\ \hline
			
		\end{tabular}
	
		\vspace{5pt}
		
		\begin{tabular}{  | p{2.6cm} | p{3.6cm} | p{2.6cm} | p{6cm} |}
				\hline
				& Space & Approx. & Comments \\ \hline\hline
				
			\multirow{4}{2.5cm}{Simplification in streaming}  
			& $O(kd\cdot\eps^{-0.5} \log^2\frac1\eps)$ & $4\sqrt{2}+\eps$ & {$2k$ vertices, continuous, \cite{ABHZ10}} \\ \cline{2-4}
			& $O(kd)$ & $8$ & {\cite{DPS19}} \\ \cline{2-4}
			& $kd\cdot O(\frac{\log\eps^{-1}}{\eps})$ & $1.22+\eps$ &  {\Cref{cor:StremingSimpConstant}} \\  \cline{2-4}
			& $k\log^2\frac1\eps\cdot O(\frac1\eps)^{\frac{d+1}{2}}$ & $1+\eps$ &  {\Cref{cor:StremingSimpEpsilon}} \\  \hline
		\end{tabular}
		\caption{\label{tbl:results} Old and new results under the discrete \frechet\ distance. We do not state the preprocessing times, as typically it is just an $m$ factor times the space bound.}
	\end{table}
	\footnotetext{Note that additional $O(\log m)$ bit operations are required in order to read the input and search the data structure.}
	
	\paragraph{Lower bound.} Driemel and Psarros \cite{DP20} proved a cell probe lower bound for decision distance oracle, providing evidence that our \Cref{thm:DOasymmetric} might be tight.
	In the cell probe model, one construct a data structure which is divided into cells of size $w$. Given a query, one can probe some cells of the data structure and preform unbounded local computation. 
	The complexity of a cell probe data structure is measured with respect to the maximum number of probes preformed during a query, and the size $w$ of the cells. \Cref{thm:DOasymmetric} works in this regime, where the number of probes is $O(1)$, and $w=O(kd)$. 
	Fix any constants $\gamma,\lambda\in(0,1)$. Consider a cell probe distance oracle $\mathcal{O}$ for curves in $\R^d$ where $d=\Theta(\log m)$, that has word size  $w<m^\lambda$, and provide answers for queries of length $k<m^\gamma$,  with approximation factor $<\sqrt{\nicefrac32}$, while using only constant number of probes. 
	Driemel and Psarros \cite{DP20} showed that $\mathcal{O}$ must use space $2^{\Omega(kd)}$.

%% file: Preliminaries.tex
	\section{Preliminaries}
		
	For two points $x,y\in\reals^d$, denote by $\|x-y\|$ the Euclidean norm.
	Let $P=(p_1,\dots,p_m)\in\reals^{d\times m}$ be a polygonal curve of length $m$ with points in $\reals^d$.
	For $1\le i\le j\le m$ denote by $P[i,j]$ the subcurve $(p_i,\dots,p_j)$, and let $P[i]=p_i$.
	We use $\circ$ to denote the concatenation of two curves or points into a new curve, for example, $P\circ P[1]=(p_1,\dots,p_m,p_1)$. Denote $[m]=\{1,\dots,m\}$.
		
	Our main goal is to solve the following problem:
	\begin{problem}[$(1+\eps)$-distance oracle]\label{problem:DO}
		Given a curve $P\in\reals^{d\times m}$, preprocess $P$ into a data structure that given a query curve $Q\in\reals^{d\times k}$ for some $k\ge 1$, returns a $(1+\eps)$ approximation of $\dfd(P,Q)$.
	\end{problem}
	We assume throughout the paper that $\eps\in(0,\frac14)$. Note that the more natural framework for \Cref{problem:DO} is when $k\le m$, however, our solution will hold for general $k$.
	
	We consider distance oracles in both the static and streaming settings. In the streaming model, the input curve $P\in \R^{d\times m}$ is presented as a data stream of a sequence of points in $\R^d$. The length $m$ of the curve is unlimited and unknown in advance, and the streaming algorithm may use some limited space $S$, which is independent of $m$. 
	The algorithm maintains a data structure that can answer queries w.r.t. the curve seen so far. In each step, a new point is reveled, and it can update the data structure accordingly. It is impossible to access previously reveled points, and the algorithm may only access the current point and the data structure.
	
	\paragraph{The discrete \frechet\ distance.}
	For the simplicity of representation, in this paper we follow the definition of \cite{EM94} and \cite{BJWYZ08} for the discrete \frechet\ distance.
	
	Consider two curves $P\in\reals^{d\times m_1}$ and $Q\in\reals^{d\times m_2}$.
	A \dfn{paired walk} along $P$ and $Q$ is a sequence of pairs $\omega=\{(\P_i,\Q_i)\}_{i=1}^t$, such that $\P_1,\dots,\P_t$ and $\Q_1,\dots,\Q_t$ partition $P$ and $Q$, respectively, into (disjoint) non-empty subcurves, and for any $i$ it holds that $|\P_i|=1$ or $|\Q_i|=1$. 
	
	A paired walk $\omega$ along $P$ and $Q$ is \dfn{one-to-many} if $|\P_i|=1$ for all $1\le i\le |P|$. We say that $\omega$ \dfn{matches} the pair $p\in P$ and $q\in Q$ if there exists $i$ such that $p\in \P_i$ and $q\in \Q_i$.
	
	The \dfn{cost} of a paired walk $\omega=\{\P_i,\Q_i\}_{i=1}^t$ along $P$ and $Q$ is $\max_{i} d(\P_i,\Q_i)$, where $d(\P_i,\Q_i)=\max_{(p,q)\in \P_i\times \Q_i} \|p-q\|_2$. In other words, it is the maximum distance over all matched pairs.
	
	The \dfn{discrete \frechet\ distance} is defined over the set $\mathcal{W}$ of all paired walks as 
	\[\dfd(P,Q)=\min_{\omega \in \mathcal{W}}\max_{(\P_i,\Q_i)\in\omega} d(\P_i,\Q_i).\]
	
	A paired walk $\omega$ is called an \dfn{optimal walk} along $P$ and $Q$ if the cost of $\omega$ is exactly $\dfd(P,Q)$. 

	\paragraph{Simplifications.}
	An \dfn{optimal $k$-simplification} of a curve $P$ is a curve $\Pi$ of length at most $k$ such that $\dfd(P,\Pi)\le\dfd(P,\Pi')$ for any other curve $\Pi'$ of length at most $k$.
	
	An \dfn{optimal $\delta$-simplification} of a curve $P$ is a curve $\Pi$ with minimum number of vertices such that $\dfd(P,\Pi)\le \delta$.
	Notice that for an optimal $\delta$-simplification $\Pi$ of a curve $P$ there always exists a an optimal walk along $\Pi$ and $P$ which is one-to-many (otherwise, we can remove vertices from $\Pi$ without increasing the distance). We will use this observation throughout the paper.
	
	The vertices of a simplification may be arbitrary, or restricted to some bounded set. A simplification $\Pi$ of $P$ is \dfn{vertex-restricted} if its set of vertices is a subset of the vertices of $P$, in the same order as they appear in $P$.
		
	In some cases, when we want to achieve reasonable space and query bounds while having a small approximation factor, we use a bi-criteria simplification.
	An \dfn{$(\alpha, k,\gamma)$-simplification} of a curve $P$ is a curve $\Pi$ of length at most $\alpha\cdot k$ such that for any curve $\Pi'$ of length at most $k$ it holds that $\dfd(P,\Pi)\le\gamma\cdot\dfd(P,\Pi')$.	
	When $\alpha=1$, we might abbreviate the notation and write $(k,\gamma)$-simplification.
		
	In our construction, we use $(k,1+\eps)$-simplifications in order to reduce the space bounds of our data structure. However, using simplification in a trivial manner leads to a constant approximation distance oracle, as follows. 
	Given a curve $P\in\reals^{d\times m}$, compute and store a $(k,1+\frac\eps2)$-simplification $\Pi$ of $P$, and for a query $Q$ compute $\dfd(\Pi,Q)$ in $O(k^2d)$ time and return $\dfd(Q,\Pi)+\dfd(\Pi,P)$.
	By the triangle inequality, $\dfd(Q,P)\le \dfd(Q,\Pi)+\dfd(\Pi,P)$. Since $\Pi$ is a $(k,1+\frac\eps2)$-simplification of $P$, we have $\dfd(\Pi,P)\le(1+\frac\eps2)\dfd(Q,P)$, and by the triangle inequality, $\dfd(Q,\Pi)+\dfd(\Pi,P)\le \dfd(Q,P)+\dfd(P,\Pi)+\dfd(\Pi,P)\le (3+\eps)\dfd(Q,P)$.
	
	\begin{observation}\label{obs:trivial_oracle}
		Given a curve $P\in\reals^{d\times m}$, there exists data structure with $O(kd)$ space, such that given a query $Q\in\reals^{d\times k}$ returns a $(3+\eps)$-approximation of $\dfd(P,Q)$ in $O(k^2d)$ time.
	\end{observation}
	
	\paragraph{Cover of a curve.}
	In order to construct an efficient distance oracle, we introduce the notion of curve cover.
	A \dfn{$(k,r,\eps)$-cover} of a curve $P\in\reals^{d\times m}$ is a set $\C$ of curves of length $k$, such that $\dfd(P,W)\le (1+\eps)r$ for every $W\in\C$, and for any curve $Q\in\reals^{d\times k}$ with $\dfd(P,Q)\le r$, there exists some curve $W\in\C$ with $\dfd(Q,W)\le \eps r$.
	
	Notice that a $(k,r,\frac\eps4)$-cover $\C$ of a curve $P$ can be used in order to construct the following decision version of a distance oracle:
	\begin{problem}[$(k,r,\eps)$-decision distance oracle]
		Given a curve $P\in\reals^{d\times m}$ and a parameters $r\in \R_+$, $\eps\in(0,\frac12)$ and $k\in [m]$, create a data structure that given a query curve $Q\in\reals^{d\times k}$, if $\dfd(P,Q)\le r$, returns a value $\Delta$ such that $\dfd(P,Q)\le\Delta\le\dfd(P,Q)+\frac{\eps}{2}r$, and if $\dfd(P,Q)>(1+\eps)r$ it returns NO. (In the case that $r<\dfd(P,Q)<(1+\eps)r$ the data structure returns either NO or a value $\Delta$ such that $\dfd(P,Q)\le\Delta\le\dfd(P,Q)+\frac{\eps}{2}r$.)
	\end{problem}
	The idea is that given a query curve $Q$, if $\dfd(P,Q)\le r$ then there exists some $W\in\C$ such that $\dfd(Q,W)\le \frac\eps4 r$ and $\dfd(P,W)\le (1+\frac\eps4)r$. By the triangle inequality $$\dfd(P,Q)\le \dfd(P,W)+\dfd(Q,W)\le \dfd(P,Q)+2\dfd(Q,W)\le \dfd(P,Q)+\frac\eps2 r.$$ On the other hand, if $\dfd(P,Q)>(1+\eps)r$, then for any $W\in\C$ we have $$\dfd(Q,W)\ge\dfd(P,Q)-\dfd(P,W)> (1+\eps)r-(1+\frac\eps4)r>\frac\eps4 r~.$$
	
	Therefore, we have the following observation.
	\begin{observation}\label{obs:cover}
		Assume that there exists a data structure that stores a $(k,r,\eps)$-cover $\C$ for $P$ of size $S$, such that given a curve $Q\in\reals^{d\times k}$ with $\dfd(Q,P)\le r$, return in time $T$ a curve $W\in\C$ with $\dfd(Q,W)\le \eps r$ and the value $\dist(W)=\dfd(P,W)$. Then there exists a $(k,r,\eps)$-decision distance oracle for $P$ with the same space and query time.
	\end{observation}
	
	Note that sometimes we abuse the notation and relate to $\C$ as the data structure from the above observation.
	
	\paragraph{Uniform grids.}
	Consider the infinite $d$-dimensional grid with edge length $\frac{\eps}{\sqrt{d}}r$, with a point at the origin.
	For a point $x\in\reals^d$, denote by $G_{\eps,r}(x,R)$ the set of grid points that are contained in $B^d_2(x,R)$, the $d$-dimensional ball of radius $R$ centered at $x$. The following claim is a generalization of Corollary 7 from \cite{FFK20}. The proof can be found in \Cref{app:grid_points_in_ball}.
	
	\begin{restatable}{claim}{gridPointsInBall}\label{clm:grid_points_in_ball}
		$|G_{\eps,r}(x,cr)|=O(\frac{c}{\eps})^d$.
	\end{restatable}

%% file: Paper_Overview.tex
	\section{Paper Overview} 

	\subsection*{\nameref*{sec:static_distance_oracle}.}
	Given a curve $P\in\reals^{d\times m}$, we first consider a more basic version of the $(1+\eps)$-distance oracle, namely, a \emph{decision distance oracle}.
	Here, in addition to $P$, we are given a distance threshold $r$. For a query curve $Q\in \reals^{d\times m}$, the decision distance oracle either returns a value $\Delta\in[\dfd(P,Q),\dfd(P,Q)+\eps r]$ or declares that $\dfd(P,Q)\ge(1+\eps)r$.
	We construct a decision distance oracle by discretizing the space of query curves (using a uniform grid).	That is, we simply store the answers to the set of all grid-curves at distance at most $(1+\eps)r$ from $P$ in a hash table. The query algorithm then ``snaps'' the points of $Q$ to the grid, to obtain the closest grid-curve, and returns the precomputed answer from the hash table. Clearly, we have a linear $O(kd)$ query time. As was shown by the authors and Katz \cite{FFK20}, the number of grid curves that we need to store is $O(\frac1\eps)^{kd}$, which is also a bound on the size of the distance oracle (see \Cref{lem:DOdecision}).

	Next, we consider a generalized version which we call a \emph{bounded range distance oracle}. Here, in addition to $P$, we are given a range of distances $[\alpha,\beta]\subset\R$. For a query $Q\in\reals^{d\times k}$, the distance oracle is guaranteed to return a $(1+\eps)$-approximation of $\dfd(P,Q)$ only if $\dfd(P,Q)\in [\alpha,\beta]$. 
	Such an oracle is constructed using $\log\frac\beta\alpha$ decision distance oracles for exponentially growing scales, and given a query we preform a binary search among them. Thus in total, compared to the decision version, we have an overhead of $\log\frac\beta\alpha$ in the space and $\log\log\frac\beta\alpha$ in the query time (see \Cref{lem:DObounded-range}).
	
	The main goal is to construct a general distance oracle that will succeed on all queries. To achieve space and query bounds independent of $m$, our first step is to precompute a $(k,1+\eps)$-simplification $\Pi$ of $P$. Note that following \Cref{obs:trivial_oracle}, given a query $Q\in\reals^{d\times k}$ we can simply return $\Delta=\dfd(Q,\Pi)+\dfd(\Pi,P)$, which is a constant approximation for $\dfd(P,Q)$ computed in $O(k^2d)$ time.
	However, we can achieve a $1+\eps$ approximation as follows,
	\begin{itemize}
		\item If $\dfd(Q,\Pi)= \Omega(\frac1\eps)\cdot\dfd(\Pi,P)$, then $\dfd(Q,\Pi)$ is a $(1+\eps)$-approximation for $\dfd(P,Q)$.
		\item If $\dfd(Q,\Pi)= O(\eps)\cdot\dfd(\Pi,P)$, then $\dfd(P,\Pi)$  is a $(1+\eps)$-approximation for $\dfd(P,Q)$.
		\item Else, we have $\dfd(Q,P)\in [\Omega(\eps),O(\frac1\eps)]\cdot \dfd(\Pi,P)$. This is a bounded range for which we can precompute a bounded range distance oracle for $P$.
	\end{itemize}
	The only caveat is that computing $\dfd(Q,\Pi)$ takes $O(k^2d)$ time. 
	Our solution is to construct a distance oracle for $\Pi$.
	At first glance, it seems that are back to the same problem. However, in this case, $\Pi$ and $Q$ have the same length!
	Thus, our entire construction boils down to computing a \emph{symmetric distance oracle}, that is, a distance oracle for the special case of $m=k$.  
	
	To achieve a near linear query time, our symmetric distance oracle first compute a coarse approximation of $\dfd(P,Q)$ in near linear time (using \Cref{thm:approximation}). Roughly speaking, if the approximated distance $\tilde{\Delta}$ is very large or very small, we show that a $(1+\eps)$-approximation can be computed directly in linear time. Else, in order to reduce the approximation factor, we maintain a polynomial number of ranges $[\alpha,\beta]$, for which we construct a bounded range distance oracles. 
	We show that if $\dfd(P,Q)$ does not fall in any of the precomputed ranges, then (approximation of) the distance can be computed in linear time.
	
	We elaborate on the different cases.
	The decision whether $\tilde{\Delta}$ is very large or very small, as well as the construction of bounded range distance oracles, are done with respect to the lengths of the edges of the input curve. First, observe that if the distance between two curves $X$ and $Y$ is smaller than half the length of the shortest edge of $X$, then $\dfd(X,Y)$ can be computed in linear time. Next, using \Cref{thm:approximation} we get a value $\tilde{\Delta}$ such that $\dfd(P,Q)\in[\frac{\tilde{\Delta}}{md},\tilde{\Delta}]$.
	Let $l_1\le l_2\le\dots\le l_{m-1}$ be a sorted list of the lengths of edges of $P$. We have four cases:
	\begin{itemize}
		\item If $\tilde{\Delta}<\frac{l_1}{2}$, then the distance between $P$ and $Q$ is smaller than half of the shortest edge in $P$, and thus  by the above observation we can compute $\dfd(P,Q)$ exactly in linear time. 
		\item If $\tilde{\Delta}>\frac{dm^2}{\eps}l_{m-1}$, then $\dfd(P[1],Q)=\max_{1\le i\le m}\|P[1]-Q[i]\|$ is a good enough approximation of $\dfd(P,Q)$, because $\dfd(P,P[1])\le m\cdot l_{m-1}<\eps\frac{\tilde{\Delta}}{md}\le \eps\cdot \dfd(P,Q)$.
	\end{itemize}
	Else, we precompute bounded range distance oracle for the ranges $[\frac{1}{\poly(\frac{md}{\eps})},\poly(\frac{md}{\eps})]\cdot l_i$ for each $i$.
	\begin{itemize}
		\item If $\tilde{\Delta}$ falls in once of the ranges above, we simply use the appropriate distance oracle to return an answer.		
		\item Else, there is some $i$ such that $\poly(\frac{md}{\eps})\cdot l_i<\tilde{\Delta}<\frac{1}{\poly(\frac{md}{\eps})}\cdot l_{i+1}$. Thus  $l_{i+1}$ is much larger than $l_i$. Let $P'$ be the curve obtain from $P$ by ``contracting'' all the edges of length at most $l_i$. It holds that  $\dfd(P',P)\le m\cdot l_i\ll \frac{\eps}{md}\cdot \tilde{\Delta}\le \eps\cdot \dfd(P,Q)$, thus  $\dfd(P',Q)$ is a $1+\eps$ approximation of  $\dfd(P,Q)$.
		From the other hand, the shortest edge of $P'$ has length at least $\frac12 l_{i+1}$ which is much larger than $\dfd(P',Q)$. Hence $\dfd(P',Q)$ can be computed in linear time.
	\end{itemize}
	
	In order to remove the logarithmic dependency on $\frac1\eps$ in the query time, we subdivide our ranges into smaller overlapping ranges, and obtain the following theorem.

	\begin{restatable}{theorem}{DOasymmetric}\label{thm:DOasymmetric}
		Given a curve $P\in\reals^{d\times m}$ and parameters $\eps\in(0,\frac14)$ and and integer $k\ge1$, there exists a distance oracle with $O(\frac{1}{\eps})^{dk}\cdot\log\eps^{-1}$ storage space, $m\log\frac{1}{\eps}\cdot\left(O(\frac{1}{\eps})^{kd}+O(d\log m)\right)$ expected preprocessing time, and $\tilde{O}(kd)$ query time.
	\end{restatable}
	
	\subsection*{\nameref*{sec:streaming_simplification}.}
	Given a curve $P$ as a stream, our goal is to maintain a $(k,1+\eps)$-simplification of $P$.	
	For	a curve $P\in\R^{d\times m}$ in the static model, an optimal $\delta$-simplification of $P$ can be computed using a greedy algorithm, that simply finds the largest index $i$ such that $P[1,i]$ can be enclosed by a ball of radius $\delta$, and then recurse for $P[i+1,m]$. This greedy simplification algorithm was presented by Bereg \etal \cite{BJWYZ08} for constant dimension, and was generalized to arbitrary dimension $d$ by the authors and Katz \cite{FFK20} (see \Cref{lem:optrsimplification}). In the static model, a $(k,1+\eps)$-simplification can be computed by searching over all the possible values of $\delta$ with the greedy $\delta$-simplification algorithm as the decision procedure. 
	
	Denote by $\gMEB$ a streaming algorithm for computing a $\gamma$-approximation of the minimum enclosing ball.
	Given a curve $P$ in a streaming fashion, and a $\gMEB$ algorithm as a black box, we first implement a streaming version of the greedy simplification algorithm called \texttt{GreedyStreamSimp} (see \Cref{alg:GreedySimp}). This algorithm gets as an input a parameter $\delta$, and acts in the same manner as the greedy simplification, where instead of a static minimum enclosing ball algorithm, it uses the $\gMEB$ black box. The resulting simplification $\Pi$ will be the sequence of centers of balls of radius $\delta$ constructed by $\gMEB$.
	Note that $\Pi$ is at distance at most $\delta$ from $P$, and every curve at distance $\delta/\gamma$ from $P$ has length at least $|\Pi|$ (see \Cref{clm:greedy_iteration}). However, the length of $\Pi$ is essentially unbounded, and our goal is to construct a simplification of length $k$.
	
	If we knew in advance the distance $\delta^*$ between $P$ and an optimal $k$-simplification of $P$, we could execute \texttt{GreedyStreamSimp} with the parameter $\delta=\gamma\delta^*$ and obtain a $(k,\gamma)$-simplification. Since $\delta^*$ is not known in advance, our \texttt{LeapingStreamSimp} algorithm tries to guess it.
	The \texttt{LeapingStreamSimp} algorithm (see \Cref{alg:StreamSimp}) gets as an input the desired length $k$, and two additional parameters $\init$ and $\inc$. It sets the initial estimation of $\delta$ to be $\init$. Then, it simply simulates \texttt{GreedyStreamSimp} (with parameter $\delta$) as long as the simplification $\Pi$ at hand is of length at most $k$. Once this condition is violated, \texttt{LeapingStreamSimp} preforms a leaping step as follows.
	Suppose that after reading $P[m]$, the length condition is violated, that is, $|\Pi|=k+1$. In this case, \texttt{LeapingStreamSimp} will increase its guess of $\delta^*$ by setting $\delta\leftarrow\delta\cdot\inc$. Then, \texttt{LeapingStreamSimp} starts a new simulation of \texttt{GreedyStreamSimp}, with the new guess $\delta$, and the previous simplification $\Pi$ as input (instead of $P[1,m]$). Now, \texttt{LeapingStreamSimp} continue processing the stream points $P[m+1,\cdots]$ as if nothing happened. Such a leaping step will be preformed each time the length condition is violated.
	As a result, eventually \texttt{LeapingStreamSimp} will hold an estimate $\delta$ and a simplification $\Pi$, such that $\Pi$ is an actual simplification constructed by the \texttt{GreedyStreamSimp} with parameter $\delta$. Alas, $\Pi$ was not constructed with respect to the observed curve $P$, but rather with respect to some other curve $P'$, such that $\dfd(P,P')\le \frac{2}{\inc}\delta$ (see \Cref{clm:one_iteration}). Furthermore, the estimate $\delta$ will be bounded by the distance to the optimal simplification $\delta^*$ multiplied by a factor of $\approx\gamma\cdot\inc$.
	
	To obtain a $1+\eps$ approximation of $\delta^*$, we run $\approx\frac1\eps$ instances of \texttt{LeapingStreamSimp}, with different initial guess parameter $\init$. Then, at each point in time, for the instance with the minimum estimation $\delta$ it holds that $\delta<(1+\eps)\delta^*$. We thus prove the following theorem.
	\begin{restatable}{theorem}{StreamSimplification}\label{thm:StreamSimplification}
		Suppose that we are given a black box streaming algorithm $\MEB_{\gamma}$ for $\gamma\in[1,2]$ which uses storage space $S(d,\gamma)$. Then for every parameters $\eps\in(0,\frac14)$ and $k\in\N$, there is a streaming algorithm which uses $O(\frac{\log\eps^{-1}}{\eps}\cdot(S(d,\gamma)+kd))$ space, and given a curve $P$ in $\R^d$ in a streaming fashion, computes a $(k,\gamma(1+\eps))$-simplification $\Pi$ of $P$, and a value $L$ such that $\dfd(\Pi,P)\le L \le \gamma(1+\eps)\delta^*$.
	\end{restatable}

	Plugging existing $\gamma-\MEB$ algorithms we obtain the followings (see \Cref{cor:StremingSimpEpsilon,cor:StremingSimpConstant}):
	\begin{itemize}
		\setlength\itemsep{0em}
		\item $(k,1+\eps)$-simplification in streaming using $O(\eps)^{-\frac{d+1}{2}}\log^{2}\eps^{-1}+O(kd\eps^{-1}\log\eps^{-1})$
		space.
		\item $(k,1.22+\eps)$-simplification in streaming using $O(\frac{\log\eps^{-1}}{\eps}\cdot kd)$
		space.
	\end{itemize}
	We note that our \texttt{LeapingStreamSimp} algorithm is a generalization of the algorithm of \cite{DPS19} for computing a $(k,8)$-simplification. Specifically, one can view the algorithm from \cite{DPS19} as a specific instance of \texttt{LeapingStreamSimp}, where fixing the parameters $\init=1$, $\inc=2$, and using a simple $2\text{-}\MEB$ algorithm.
		
	\subsection*{\nameref*{sec:streaming_distance_oracle}.}
	Our basic approach here imitating our static distance oracle from \Cref{thm:DOasymmetric}.
	We maintain a $(k,1+\eps)$-simplification $\Pi$ of $P$ (using \Cref{cor:StremingSimpEpsilon}). As we have the simplification explicitly, we can also construct a symmetric-distance oracle for $\Pi$. Thus, when a query $Q$ for $P$ arrives, we can estimate $\dfd(\Pi,Q)$ quickly.
	If either $\dfd(\Pi,Q)>\frac1\eps\cdot \dfd(\Pi,P)$ or $\dfd(\Pi,Q)<\eps\cdot \dfd(\Pi,P)$, as previously discussed, we can answer immediately.
	Else, we have $\dfd(Q,P)\in [\Omega(\eps),O(\frac1\eps)]\cdot \dfd(\Pi,P)$, which is a bounded range.
	In the static case, we simply prepared ahead answers to all the possible queries in this range.
	However, in the streaming case, this range is constantly changing, and is unknown in advance. How can we be prepared for the unknown?
	
	The first key observation is that given a parameter $r$, one can maintain a decision distance oracle \footnote{Actually by decision distance oracle here we mean cover, see \Cref{obs:cover}.} in a stream. Specifically, given a decision distance oracle for a curve $P[1,m]$ with storage space independent of $m$ (as in \Cref{lem:StramDOdecision}), and a new point $P[m+1]$, we show how to construct a decision distance oracle for $P[1,m+1]$.
	However, the scales $r$ for which we construct the decision distance oracles are unknown in advance, and we need a way to update $r$ on demand.
	
	Our solution is similar in spirit to the maintenance of simplification in the stream.
	That is, we will create a leaping version of the decision oracle, i.e., a data structure that receives as input a pair of parameters $\init$ and $\inc$. Initially it sets the scale parameter $r$ to $\init$. As long as the distance oracle is not empty (i.e. there is at least one curve at distance $r$ from $P$), it continues simulating the streaming algorithm that construct a decision distance oracle for fixed $r$.
	If it becomes empty after reading the point $P[m]$ for the stream, then instead of despairing, the oracle updates its scale parameter to $r\cdot \inc$, choose an arbitrary curve $W$ from the distance oracle of $P[1,m-1]$, and initialize a new distance oracle for $W\circ P[m]$ using the new parameter $r$. From here on, the oracle continue simulating the construction of a decision distance oracle as before, while preforming a leaping step each time it becomes empty.
	
	As a result, at each step we have an actual decision distance oracle for some parameter $r$. Alas, the oracle returns answers not with respect to the observed curve $P$, but rather with respect to some other curve $P'$, such that $\dfd(P,P')\le \frac{2}{\inc}r$ (see \Cref{lem:leapingStreaming}). 
	We show that if we maintain $O(\log \frac1\eps)$ such leaping distance oracles for different values of $\init$, we will always be able to answer queries for curves $Q$ such that $\dfd(Q,P)\in [\Omega(\eps),O(\frac1\eps)]\cdot \dfd(\Pi,P)$. 
	
	\begin{restatable}{theorem}{mainStreamingDO}\label{thm:mainStreamingDO}
		Given parameters $\eps\in(0,\frac14)$ and $k\in \N$, there is a streaming algorithm that uses $O(\frac{1}{\eps})^{kd}\log\eps^{-1}$ space, and given a curve $P$ with points in $\R^d$, constructs a $(1+\eps)$-distance oracle with $\tilde{O}(kd)$ query time.
	\end{restatable}

	\subsection*{\nameref*{sec:subcurve_distance_oracle}.} Following \cite{DH13} and \cite{Filtser18}, we consider a generalization of the distance oracle problem, where the query algorithm gets as an input two index $1\le i\le j\le m$ in addition to a query curve $Q\in \reals^{d\times k}$, and return $(1+\eps)$-approximation of $\dfd(P[i,j],Q)$.
	Note that a trivial solution is storing $O(m^2)$ distance oracles: for any $1\le i\le j\le m$ store a $(1+\eps)$-distance oracle for $P[i,j]$. However, when $m$ is large, one might wish to reduce the quadratic storage space at the cost of increasing the query time or approximation factor.
	
	Before presenting our solution to the above problem, we introduce a closely related problem which we call the ``zoom-in'' problem. Given a curve $P\in \reals^{d\times m}$ and an integer $1\le k< m$, our goal is to preprocess $P$ into a data structure that given $1\le i< j\le m$, return an $(\alpha,k,\gamma)$-simplification of $P[i,j]$.
	Our solutions to the zoom-in problem and the distance oracle to a subcurve problem have a similar basic structure, which consists of hierarchically partitioning the input curve $P$. Given a query, a solution is constructed by basically concatenating two precomputed solutions. 
	We obtain the following theorems.
	\begin{restatable}{theorem}{zoomin}\label{thm:zoomin2k}
		Given a curve $P$ consisting of $m$ points and parameters $k\in[m]$ and $\eps\in(0,\frac12)$ 
		there exists a data structure with $O(mkd\log\frac{m}{k})$ space, such that given a pair of indices $1\le i< j\le m$, returns in $O(kd)$ time an $(k,1+\eps,2)$-simplification of $P[i,j]$.
		The prepossessing time for general $d$ is $\tilde{O}(m^{2}d\eps^{-4.5})$, while for fixed $d$ is $\tilde{O}(m^{2}\eps^{-1})$. 
	\end{restatable}
	
	\begin{restatable}{theorem}{DOsubCurve}\label{thm:DOsubCurve}
		Given a curve $P\in \reals^{d\times m}$ and parameter $\eps>0$, 
		there exists a a data structure that given a query curve $Q\in \reals^{d\times k}$, and two indexes $1\le i\le j\le m$, returns an $(1+\eps)$-approximation of $\dfd(P[i,j],Q)$. The data structure has
		$m\log m\cdot O(\frac{1}{\eps})^{dk}\cdot\log\eps^{-1}$ storage space, $m^{2}\log\frac{1}{\eps}\cdot\left(O(\frac{1}{\eps})^{kd}+O(d\log m)\right)$ expected preprocessing time, and $\tilde{O}(k^2d)$ query time.
	\end{restatable}

	\subsection*{\nameref*{sec:high_d_optimization}.}
	In \Cref{lem:allDistances} we present a simple technique which is useful when one wants to get an approximated distance over a set of points in any dimension $d$.
	We use it to remove the exponential factor from the approximation algorithm of \cite{CR18}, and to generalize the algorithm of \cite{BJWYZ08} for computing a $(1+\eps)$-approximation of the optimal $k$ simplification in any dimension. Note that the algorithm of \cite{BJWYZ08} has running time $\tilde{O}(mk)$ for a curve in $P\in\R^{d\times m}$, and by considering $(k,1+\eps)$-simplifications instead of optimal $k$-simplifications we manage to reduce the running time to $\tilde{O}(\frac{md}{\eps^{4.5}})$.
	We obtain the following theorems.	
	
	\begin{restatable}{theorem}{CrudeFreshetApproxHighDim}\label{thm:approximation}
		Given two curves $P$ and $Q$ in $\reals^{d\times m}$, and a value $f\ge 1$, there is an algorithm that returns in $O\left(md\log(md)\log d+(md/f)^{2}d\log(md)\right)=\tilde{O}(md+(md/f)^{2}d)$\textsl{}
		time a value $\tilde{\Delta}$ such that $\dfd(P,Q)\le\tilde{\Delta}\le f\cdot\dfd(P,Q)$.
	\end{restatable}

	\begin{restatable}{theorem}{SimplificationHighDim}\label{thm:SimplificationHighDim}
		Given a curve $P\in\R^{m\times d}$ and parameters $k\in[m]$, $\eps\in(0,\frac12)$, there is an $\tilde{O}(\frac{md}{\eps^{4.5}})$-time algorithm that computes a $(k,1+\eps)$-simplification $\Pi$ of $P$.
		In addition the algorithm returns a value $\delta$ such that $\dfd(P,\Pi)\le\delta\le(1+\eps)\delta^*$, where $\delta^*$ is the distance between $P$ to an optimal $k$-simplification.\\
		Furthermore, if $d$ is fixed, the algorithm can be executed in $m\cdot O(\frac{1}{\eps}+\log\frac{m}{\eps}\log m)$ time.
	\end{restatable}

%% file: Static_Distance_Oracle.tex
\section{Distance Oracle: the static case}	\label{sec:static_distance_oracle}
We begin by constructing a $(1+\eps)$-distance oracle for the static case, where a curve $P\in\reals^{d\times m}$ is given in the preprocessing stage.
To achieve a $(1+\eps)$ approximation for the distance between $P$ and a query $Q\in\reals^{d\times k}$ in near linear time, our distance oracle first computes a very rough estimation of this distance. Then, in order to reduce the approximation factor, we maintain a polynomial number of ranges $[\alpha,\beta]$, for which we store a distance oracle that can answer queries only when the answer is in the range $[\alpha,\beta]$. This structure uses a set of $(k,r,\eps)$-covers, where $r$ grows exponentially in the given range $[\alpha,\beta]$.

We describe the ingredients of our distance oracle from the bottom up, starting with the basic construction of a curve cover, then present the bounded range distance oracle, describe a solution for the case where $k=m$ (the symmetric case), and finally show how to combine all the ingredients and construct a $(1+\eps)$-distance oracle for $P$ with near linear query time and $O(\frac1\eps)^{kd}\log\eps^{-1}$ storage space.

\subsection{Cover of a curve}\label{subsec:StaticDecisionDO}
Given input curve $P\in\reals^{d\times m}$, in this section we show how to construct a data structure that stores a $(k,r,\eps)$-cover $\C$ of size $O(\frac1\eps)^{kd}$ for $P$, and has a linear look-up time.

Our data structure is based on the ANN data structure presented by Filtser et al. \cite{FFK20}. For a single curve $P$, this data structure essentially solves a decision version of the distance oracle: given parameters $r$ and $\eps$, the $ (r,1+\eps)$-ANN data structure uses $O(\frac{1}{\eps})^{kd}$ storage space, and given a query curve $Q\in\reals^{d\times k}$ returns YES if $\dfd(P,Q)\le r$ and NO if  $\dfd(P,Q)>(1+\eps)r$ (if $r<\dfd(P,Q)\le (1+\eps)r$ it can return either YES or NO).

Using the same technique from \cite{FFK20} with a slight adaptation, one can construct a $(k,r,\eps)$-cover with the same space and look-up bounds. We include the basic details here for completeness.

Consider the infinite $d$-dimensional grid with edge length $\frac{\eps}{\sqrt{d}}r$, and
let 
\[
\G=\bigcup_{1\le i\le m} G_{\eps,r}(P[i],(1+\eps)r).
\] 
Let $\C$ be the set of all curves $W$ with $k$ points from $\G$, such that $\dfd(P,W)\le (1+\eps)r$. Filtser \etal \cite{FFK20} showed that $|\C|=O(\frac{1}{\eps})^{kd}$, and that it can be computed in $m\cdot O(\frac{1}{\eps})^{kd}$ time.

\paragraph{The data structure.} We insert the curves of $\C$ into the dictionary $\D$ as follows. For each curve $W\in \C$, if $W\notin \D$, insert $W$ into $\D$, and set $\dist(W)\leftarrow \dfd(P,W)$.

Filtser \etal \cite{FFK20} showed that $\D$ can be implemented using Cuckoo Hashing~\cite{PR04}, so that given a query curve $Q$, one can find $Q$ in $\D$ (if it exists) in $O(kd)$ time, the storage space required for $\D$ is $O(\frac{1}{\eps})^{kd}$, and it can be constructed in $m\cdot (O(\frac{1}{\eps})^{kd}+d\log m)$ expected time.

\paragraph{The query algorithm.} Let $Q\in\reals^{d\times k}$ be the query curve. The query algorithm is as follows: For each $1\le i\le k$ find the grid point $x_i$ (not necessarily from $\G$) closest to $Q[i]$. This can be done in $O(kd)$ time by rounding. Then, search for the curve $W'=(x_1,\dots,x_k)$ in the dictionary $\D$. If $W'$ is in $\D$, return $W'$ and $\dist(W')$, otherwise, return NO. The total query time is then $O(kd)$.

\paragraph{Correctness.} First, by the construction, for any $W\in\C$ we have $\dfd(P,W)\le (1+\eps)r$. Secondly, let $Q\in\reals^{d\times k}$ be a query curve such that that $\dfd(P,Q)\le r$. Notice that $\left\Vert Q[i]-x_i\right\Vert_{2}\le \frac{\eps}{2\sqrt{d}} r$ because the length of the grid edges is $\frac{\eps}{\sqrt{d}}r$, and thus $\dfd(Q,W')\le \frac\eps2 r$.
By the triangle inequality, $\dfd(P,W')\le\dfd(P,Q)+\dfd(Q,W')\le (1+\eps)r$,
and therefore $W'$ is in $\C$.

By \Cref{obs:cover} we obtain the following lemma.
\begin{lemma}\label{lem:DOdecision}
	Given a curve $P\in\reals^{d\times m}$ and a parameters $r\in \R_+$, $\eps\in(0,\frac14)$ and $k\ge 1$, there is an algorithm that constructs a $(k,r,\eps)$-decision distance oracle with $O(\frac{1}{\eps})^{kd}$ storage space, $m\cdot\left(O(\frac{1}{\eps})^{kd}+O(d\log m)\right)$ expected preprocessing time, and $O(kd)$ query time.
\end{lemma}

\subsection{Bounded range distance oracle}
We now show how to use $(k,\eps,r)$-covers in order to solve the following problem.

\begin{problem}[Bounded range distance oracle]
	Given a curve $P\in\reals^{d\times m}$, a range $[\alpha,\beta]$ (where $\beta\ge 4\alpha$), and parameters $\eps\in(0,\frac14)$ and $k\ge 1$, preprocess $P$ into a data structure that given a query curve $Q\in\reals^{d\times k}$ with $\dfd(P,Q)\in[\alpha,\beta]$, returns a $(1+\eps)$ approximation of $\dfd(P,Q)$.
\end{problem}

\paragraph{The data structure.}
For every $0\le i\le \lceil\log\nicefrac{\beta}{\alpha}\rceil$, we construct a decision distance oracle $\D_i$ with parameter $r_i=\alpha\cdot 2^i$, and $\eps'=\frac\eps4$. 
The total storage space is therefore $O(\log\nicefrac{\beta}{\alpha})\cdot O(\frac{1}{\eps})^{kd}$, while the preprocessing time is $m\log\nicefrac{\beta}{\alpha}\cdot\left(O(\frac{1}{\eps})^{kd}+O(d\log m)\right)$ in expectation.

\paragraph{The query algorithm.} Given a query curve $Q\in\reals^{d\times k}$ such that $\dfd(P,Q)\in[\alpha,\beta]$, preform a binary search on the values $r_i=\alpha\cdot 2^i$ for $0\le i \le \lceil\log\nicefrac{\beta}{\alpha}\rceil$, using the decision distance oracles as follows. Let $[s',t']$ be the current range. We describe a recursive binary search where the invariant is that $\alpha\cdot 2^{s'}\le\dfd(P,Q)\le\alpha\cdot 2^{t'}$. If $t'\le s'+2$ return the answer of $\D_{t'}$. Else, let $x=\lfloor \frac{t'-s'}{2}\rfloor$ and query $\D_x$ with $Q$. If it returns a distance, then set the current range to $[s',x+1]$. Otherwise set current range to $[x,t']$.

The number of decision queries is $O(\log \log \nicefrac{\beta}{\alpha})$.
By \Cref{lem:DOdecision}, the total query time is therefore $O(kd\log \log \nicefrac{\beta}{\alpha})$.

\paragraph{Correctness.}
Assume that $\alpha\cdot 2^{s'}\le\dfd(P,Q)\le\alpha\cdot 2^{t'}$. If  $t'\le s'+2$ then $\D_{t'}$ returns a value $\Delta$ such that $\dfd(P,Q)\le\Delta\le \dfd(P,Q)+\eps' r_{t'}=\dfd(P,Q)+\frac{\eps}{4}\alpha\cdot 2^{t'}=\dfd(P,Q)+\eps\alpha\cdot 2^{s'}\le (1+\eps)\dfd(P,Q)$.

Else, if $\D_x$ returns a distance value, then $\dfd(P,Q)\le(1+\eps)r_x=(1+\eps)\alpha\cdot 2^x\le \alpha\cdot 2^{x+1}$ (for $\eps<1$), so $\alpha\cdot 2^{s'}\le\dfd(P,Q)\le\alpha\cdot 2^{x+1}$ and the invariant still hold. Also, notice that $x+1<t'$ because $t'-s'>1$. If $\D_x$ returns NO, then $\dfd(P,Q)>r_x=\alpha\cdot 2^x$, so $\alpha\cdot 2^{x}\le\dfd(P,Q)\le \alpha\cdot 2^{t'}$ and the invariant still hold.

We obtain the following lemma.
\begin{lemma}\label{lem:DObounded-range}
		Given a curve $P\in\reals^{d\times m}$ a range $[\alpha,\beta]$ where $\beta\ge4\alpha$, and parameters $\eps\in(0,\frac14)$, $k\ge 1$, there exists a bounded range distance oracle with $O(\frac{1}{\eps})^{kd}\cdot \log\nicefrac{\beta}{\alpha}$ storage space, $m\log\nicefrac{\beta}{\alpha}\cdot\left(O(\frac{1}{\eps})^{kd}+O(d\log m)\right)$
	expected preprocessing time, and $O(kd\log\log \nicefrac{\beta}{\alpha})$ query time.
\end{lemma}

\subsection{Symmetric distance oracle}
We construct a $(1+\eps)$-distance oracle for the symmetric case of $k=m$.
This time, the query algorithm of our distance oracle does not simply return a precomputed value, but actually preform a smart case analysis which allows both fast query time and a relatively small storage space complexity. The main idea is to first preform a fast computation of a very rough approximation of the distance between the query $Q$ and the input $P$. If the approximated distance is very large or very small, we show that a $(1+\eps)$-approximation can be returned right away. For the other cases, we use a (precomputed) set of bounded range distance oracles, as described in the previous section.

The decision whether the approximated distance is very large or very small depends on the length of the smallest and largest edges of $P$.
Denote ${\lambda(P)=\frac{1}{2}\min_{1\le i\le m-1}\{\Vert P[i]-P[i+1]\Vert\}}$, i.e. $\lambda(P)$ is half the length of the shortest edge in $P$.
Let $l_1\le l_2\le\dots\le l_{m-1}$ be a sorted list of the lengths of $P$'s edges, and set $l_0=\lambda(P)=\frac{l_1}{2}$ and $l_m=\frac{dm^2}{\eps}l_{m-1}$.

Let $X\in\reals^{d\times m_1}$ and $Y\in\reals^{d\times m_2}$ be two curves. Notice that there are no $i\in[m_1-1]$ and $j\in[m_2]$ such that $\Vert Y[j]-X[i]\Vert<\lambda(X)$ and $\Vert Y[j]-X[i+1]\Vert<\lambda(X)$, because otherwise $\Vert X[i]-X[i+1]\Vert<2\lambda(X)=\min_{1\le i\le m_1-1}\{\Vert X[i]-X[i+1]\Vert\}$. Therefore, if $\dfd(X,Y)<\lambda(X)$, then there exists a single paired walk $\omega$ along $X$ and $Y$ with cost $\dfd(X,Y)$, and $\omega$ is one-to-many.

\begin{algorithm}[h]
	\caption{\texttt{SmallDistance$(X,Y)$}}\label{alg:SmallDistance}
	\DontPrintSemicolon
	\SetKwInOut{Input}{input}\SetKwInOut{Output}{output}
	\Input{Curves $X\in\reals^{d\times m_1}$ and $Y\in\reals^{d\times m_2}$}
	\Output{Either $\dfd(X,Y)$ or NO}
	\BlankLine
	Compute $\lambda(X)$\;
	Set $j\leftarrow 1$, $\Delta\leftarrow 0$\;
	\For{$1\le i \le m_1$}{
		\If{$j> m_2$ OR $\Vert X[i]-Y[j] \Vert \ge \lambda(X)$}
		{\Return NO \;}
		\While{$\Vert X[i]-Y[j] \Vert < \lambda(X)$}{
			$\Delta\leftarrow\max\{\Delta,\Vert X[i]-Y[j] \Vert\}$\; $j\leftarrow j+1$\;
		}
	}
	\Return $\Delta$\;
\end{algorithm}

Consider \Cref{alg:SmallDistance}, which essentially attempts to compute the one-to-many paired walk $\omega$ along $X$ and $Y$ with respect to $\lambda(X)$, in a greedy fashion. If the algorithm fails to do so, then $\dfd(X,Y)\ge\lambda(X)$. It is easy to see that the running time of \Cref{alg:SmallDistance} is $O(m_1+m_2)$. Therefore, we obtain the following claim.
\begin{claim}\label{clm:long_edges}
	\Cref{alg:SmallDistance} runs in linear time, and if $\dfd(X,Y)<\lambda(X)$ then it returns $\dfd(X,Y)$, else, it returns NO.
\end{claim}

Our query algorithm first computes an $m$-approximation $\tilde{\Delta}$ of $\dfd(P,Q)$ using \Cref{thm:approximation}. 
The query algorithm contains four basic cases, depending on the value $\tilde{\Delta}$. Cases 1-3 do not require any precomputed values, and compute the returned approximated distance in linear time. For case 4, we store a set of $O(m\cdot\lceil\log_m\frac{1}{\eps}\rceil))$ bounded range distance oracles as follows. 

First, consider the following ranges of distances: for $1\le i\le m-1$, set $[\alpha_i,\beta_i]=[\frac{1}{5dm}l_i,\frac{dm^2}{\eps}l_i]$.
Notice that $\frac{\beta_i}{\alpha_i}=\frac{5d^2m^3}{\eps}$.
Next, for each of the above ranges we construct a set of overlapping subranges, each with ratio $(dm)^2$. More precisely, for every $0\le i\le m-1$ and $\left\lfloor \log_{dm}\frac{1}{5dm}\right\rfloor \le j\le\left\lfloor \log_{dm}\frac{m}{\epsilon}\right\rfloor$, set $[\alpha^j_i,\beta^j_i]=[l_i\cdot (dm)^j,l_i\cdot (dm)^{j+2}]$ and construct a bounded range distance oracle $\D^j_i$ with the range $[\alpha^j_i,\beta^j_i]$ using \Cref{lem:DObounded-range}. \footnote{Note that initially we could use \Cref{lem:DObounded-range} directly on the range $[\beta_i,\alpha_i]$. However, the further subdivision to sub ranges of size polynomial in $m$ saves an $\log\log\frac1\eps$ factor from the query time.}

\paragraph{The query algorithm.} Given a query curve $Q\in\reals^{d\times m}$, compute in $O(md\log(md)\log d)$ time a value $\tilde{\Delta}$ such that $\dfd(P,Q)\le\tilde{\Delta}\le md\cdot\dfd(P,Q)$ (using the algorithm from \Cref{thm:approximation}).
\begin{description}
	\item[Case 1:] $\tilde{\Delta}<\frac{l_1}{2}$.\qquad Return \texttt{SmallDistance$(P,Q)$}.
	\item[Case 2:] $\tilde{\Delta}>\frac{dm^2}{\eps}\cdot l_{m-1}$.\qquad Return $\dfd(P[1],Q)$.		
	\item If both cases 1 and 2 do not hold, then we have $l_0=\frac{l_1}{2}\le\tilde{\Delta}\le\frac{dm^2}{\eps}\cdot l_{m-1}=l_m$. There must be an index $i\in[1,m-1]$ such that one of the following two cases hold.
	\item[Case 3:]  $\frac{dm^2}{\eps}l_i\le\tilde{\Delta}\le \frac{1}{5}l_{i+1}$.\qquad 
	Let $S=\{P[j] \mid \Vert P[j]-P[j+1]\Vert \le l_i \}$, and let $P'$ be the curve obtained by removing the points of $S$ from $P$.
	Return \texttt{SmallDistance$(P',Q)$}.
	\item[Case 4:] $\frac{1}{5}l_i\le\tilde{\Delta}\le \frac{dm^2}{\eps}l_i$.\qquad 
	We have $\dfd(P,Q)\in[\frac{\tilde{\Delta}}{dm},\tilde{\Delta}]$, so let $j$ be an index such that $[\frac{\tilde{\Delta}}{dm},\tilde{\Delta}]\subseteq[\alpha^j_{i},\beta^j_{i}]$, query $\D^j_{i}$ with $Q$ and return the answer.	
\end{description}

\paragraph{Correctness.} We show that the query algorithm returns a value $\Delta^*$ such that $(1-\eps)\dfd(P,Q)\le\Delta^*\le(1+\eps)\dfd(P,Q)$. 

First, we claim that the four cases in our query algorithm are disjoint, and that $\tilde{\Delta}$ falls in one of them. The reason is that if cases 1-2 do not hold, then $l_0\le\tilde{\Delta}\le l_m$, so there must exists an index $i$ such that $ l_i\le \tilde{\Delta}\le l_{i+1}$. If $\frac{m^2}{\eps}l_i\le\tilde{\Delta}\le \frac{1}{5}l_{i+1}$ then we are in case 3, and otherwise, $l_i\le\tilde{\Delta}< \frac{dm^2}{\eps}l_i$ or $\frac{1}{5}l_{i+1}<\tilde{\Delta}\le l_{i+1}$ and we are in case 4.

We proceed by case analysis.
\begin{description}
	\item[Case 1:] $\tilde{\Delta}<\frac{l_1}{2}$.\qquad Then $\dfd(P,Q)\le\tilde{\Delta}\le\lambda(P)$, and by \Cref{clm:long_edges} we return $\dfd(P,Q)$.
	\item[Case 2:] $\tilde{\Delta}>\frac{dm^2}{\eps}\cdot l_{m-1}$.\qquad Then $\dfd(P,Q)\ge\frac{1}{md}\cdot\tilde{\Delta}>\frac{m}{\eps}\cdot l_{m-1}$. Notice that $\dfd(P[1],P)\le\sum_{i=1}^{m-1}l_i\le m\cdot\l_{m-1}<\eps\cdot\dfd(P,Q)$. By the triangle inequality, we have
	\[
	\dfd(P[1],Q)\le\dfd(P,Q)+\dfd(P[1],P)<(1+\eps)\dfd(P,Q),
	\] and
	\[
	\dfd(P[1],Q)\ge\dfd(P,Q)-\dfd(P[1],P)>(1-\eps)\dfd(P,Q).
	\]
	\item[Case 3:] $\frac{dm^2}{\eps}l_i\le\tilde{\Delta}\le \frac{1}{5}l_{i+1}$.\qquad Thus 
	\begin{equation}
	\frac{m}{\eps}l_{i}\le\frac{\tilde{\Delta}}{dm}\le\dfd(P,Q)\le\tilde{\Delta}\le\frac{1}{5}l_{i+1}~.\label{eq:Case3SymmetricStatic}
	\end{equation}
	
	Denote $P'=(P[j_1],P[j_2],\dots,P[j_u])$. We argue that $\lambda(P')> \frac14 l_{i+1}$, that is, for every $s\in [1,u-1]$, $\Vert P[j_s]-P[j_{s+1}]\Vert> \frac12 l_{i+1}$.
	Fix such an index $s$. If $j_s=j_{s+1}-1$, then as $P[j_s]\notin S$, $\Vert P[j_s]-P[j_{s+1}]\Vert\ge l_{i+1}$.
	Otherwise, as $P[j_s]\notin S$ and $\{{P[j_{s}+1]},\dots,{P[j_{s+1}-1]}\}\subseteq S$, by the triangle inequality,
	\begin{align*}
	\Vert P[j_{s}]-P[j_{s+1}]\Vert & \ge\Vert P[j_{s}]-P[j_{s}+1]\Vert-\sum_{t=1}^{j_{s+1}-j_{s}-1}\Vert P[j_{s}+t]-P[j_{s}+t+1]\Vert\\
	& \ge l_{i+1}-m\cdot l_{i}\overset{(\ref{eq:Case3SymmetricStatic})}{\ge}l_{i+1}-\frac{\eps}{5}\cdot l_{i+1}>\frac{1}{2}l_{i+1}~.
	\end{align*}
	
	Notice that $\dfd(P,P')\le m\cdot l_i\overset{(\ref{eq:Case3SymmetricStatic})}{\le}\eps\cdot \dfd(P,Q)$. This is as the cost of the paired walk $\omega=\{(P[1,j_1],P'[1])\}\cup\{(P[j_{s-1}+1,j_s],P'[s])\mid 2\le s\le u\}$ is at most $m\cdot l_i$ (again using triangle inequality). Thus  $\dfd(P',Q)\le\dfd(P,Q)+\dfd(P,P')\le(1+\eps)\dfd(P,Q)$ and $\dfd(P',Q)\ge\dfd(P,Q)-\dfd(P,P')\ge(1-\eps)\dfd(P,Q)$.
	
	Finally, as $\dfd(P',Q)\le(1+\eps)\dfd(P,Q)\overset{(\ref{eq:Case3SymmetricStatic})}{\le}\frac{1+\eps}{5}l_{i+1}<\frac{1}{4}l_{i+1}<\lambda(P')$, by \Cref{clm:long_edges} we can compute $\dfd(P',Q)$, which is a $(1+\eps)$-approximation of $\dfd(P,Q)$.
	
	\item[Case 4:] $\frac{1}{5}l_i\le\tilde{\Delta}\le \frac{dm^2}{\eps}l_i$.\qquad
	Then $\frac{1}{5dm}l_{i}\le\frac{\tilde{\Delta}}{dm}\le\dfd(P,Q)\le\tilde{\Delta}\le\frac{dm^{2}}{\eps}l_{i}$.			
	Let $\left\lfloor \log_{dm}\frac{1}{5dm}\right\rfloor \le j\le\left\lfloor \log_{dm}\frac{m}{\epsilon}\right\rfloor$
	 be an index such that
	$[\frac{\tilde{\Delta}}{dm},\tilde{\Delta}]\subseteq[l_{i}\cdot(dm)^{j},l_{i}\cdot(dm)^{j+2}]$. For example the maximal $j$ such that $l_{i}\cdot(dm)^{j}\le\frac{\tilde{\Delta}}{dm}$
	will do.
	By \Cref{lem:DObounded-range}, as $\dfd(P,Q)$ is in the range, using the bounded range distance oracle $\D^j_i$ we will return an $(1+\eps)$-approximation of $\dfd(P,Q)$.
	
\end{description}
\paragraph{Running time and storage space.} Computing $\tilde{\Delta}$ takes $O(md\log(md)\log d)$ time according to \Cref{thm:approximation}.
Deciding which case is relevant for our query takes $O(m)$ time. 
Case 1 takes $O(md)$ time according to \Cref{clm:long_edges}. Case 2 can be computed in $O(md)$ times, as it is simply finding the maximum among $m$ distances, each computed in $O(d)$ time. For case 3, we can compute $P'$ sequentially in $O(md)$ time, and then compute $\dfd(P',Q)$ in $O(md)$ time using \Cref{clm:long_edges}.
For case 4, according to \Cref{lem:DObounded-range}, the query time is $O(md\log\log md)$ time. 
Thus $O(md\log(md)\log d)$  in total.

\sloppy Cases 1-3 do not require any precomputed values, while in case 4 each $\D^j_i$ uses $O(\frac{1}{\eps})^{md}\cdot\log md=O(\frac{1}{\eps})^{md}$ space and $m\log md\cdot\left(O(\frac{1}{\eps})^{md}+O(d\log m)\right)=O(\frac{1}{\eps})^{md}$
expected preprocessing time.
Thus the total storage space and running time is thus  $m\cdot\log_{dm}(\frac{5dm^{2}}{\eps})\cdot O(\frac{1}{\eps})^{md}=O(\frac{1}{\eps})^{md}\cdot\log\eps^{-1}$. 
We conclude,
\begin{theorem}\label{thm:DOsymmetric}
	Given a curve $P\in\reals^{d\times m}$ and parameter $\eps\in(0,\frac14)$, there exists a $(1+\eps)$-distance oracle with $O(\frac{1}{\eps})^{dm}\cdot\log\eps^{-1}$ storage space and  preprocessing time, and $\tilde{O}(md)$ query time.
\end{theorem}

\begin{remark}
	Interestingly, \cite{FFK20} constructed a near neighbor data structure with query time $O(md)$. Combined with the standard reduction \cite{HIM12}, they obtain a nearest neighbor data structure with query time $O(md\log n)$ (given that the number of input curves is $n$).   
	The case of distance oracle is similar with $n=1$. However, we cannot use the NNS data structure as a black box, as it is actually returns a neighbor and not a distance. Obtaining a truly linear query time ($O(md)$) in \Cref{thm:DOsymmetric} is an intriguing open question.
\end{remark}

\subsection{$(1+\eps)$-distance oracle for any $k$}\label{subs:static_asymmetric}
Let $\Pi$ be a $(k,1+\eps)$-simplification of $P$, which can be computed in $\tilde{O}(\frac{md}{\eps^{4.5}})$ time using \Cref{thm:SimplificationHighDim}.~\footnote{If $d\le 4$, then the running time will be $\tilde{O}(\frac m\eps)$. In any case, the contribution of this step to the preprocessing time is insignificant.}
In addition we obtain from \Cref{thm:SimplificationHighDim} an estimate $L$ such that $\dfd(P,\Pi)\le L\le (1+\eps)\dfd(P,\Pi^*)$, where $\Pi^*$ is the optimal $k$-simplification.

We construct a symmetric distance oracle $\mathcal{O}_\Pi$ for $\Pi$ using \Cref{thm:DOsymmetric}.
In addition, for every index $i\in[0,\left\lceil\log\frac1\eps\right\rceil]$ we construct an asymmetric bounded range distance oracle $\mathcal{O}_i$ for $P$, with the range $[2^{i-1}\cdot L,2^{i+3} \cdot L]$ using \Cref{lem:DObounded-range}. \footnote{Actually, as the aspect ratio is constant, we could equivalently used here \Cref{lem:DOdecision} with parameters $2^{i+3}\cdot L$ and $\frac{\eps}{16}$ instead of \Cref{lem:DObounded-range}.}

\paragraph{The query algorithm.}
Given a query curve $Q\in\reals^{d\times k}$, we query $\mathcal{O}_\Pi$ and get a value $\Delta$ such that $\dfd(Q,\Pi)\le\Delta\le(1+\eps)\dfd(Q,\Pi)$.
\begin{itemize}
	\item If $\Delta\ge\frac{1}{\eps}\cdot L$, return $(1+\eps)\Delta$.
	\item Else, if $\Delta\le 7L$, return the answer of $\mathcal{O}_0$ for $Q$.
	\item Else, let $i$ be the maximal index such that $2^i\cdot L\le \frac\Delta2$, and return the answer of $\mathcal{O}_i$ for $Q$.
\end{itemize}

\paragraph{Correctness.} We show that the query algorithm always return a value $\tilde{\Delta}$ such that $\dfd(P,Q)\le \tilde{\Delta}\le (1+6\eps)\dfd(P,Q)$. Afterwards, the $\eps$ parameter can be adjusted accordingly. By triangle inequality it holds that
\[
\frac{\Delta}{1+\epsilon}-L\le\dfd(Q,\Pi)-\dfd(P,\Pi)\le\dfd(P,Q)\le\dfd(Q,\Pi)+\dfd(P,\Pi)\le\Delta+L~.
\]
\begin{itemize}
	\item If $\Delta\ge\frac{1}{\eps}\cdot L$, then $\dfd(P,Q)\le(1+\eps)\Delta$. Furthermore, it holds that $\dfd(P,Q) \ge\frac{1}{1+\eps}\Delta-\eps\Delta\ge(1-2\eps)\Delta$, which implies $(1+\eps)\Delta\le\frac{1+\eps}{1-2\eps}\cdot\dfd(P,Q)<(1+6\eps)\cdot\dfd(P,Q)$.		
	
	\item Else, if $\Delta\le 7L$ then $\dfd(P,Q)\le 8L$. Since $Q$ is a $k$ point curve, we have that $\dfd(P,Q)\ge\dfd(P,\Pi^{*})\ge\frac{L}{1+\eps}>\frac{L}{2}$
	(here $\Pi^{*}$ is the optimal $k$-simplification).
	Hence $\dfd(P,Q)\in [\frac{L}{2},8L]$, and $\mathcal{O}_0$ returns a $(1+\eps)$-approximation for $\dfd(P,Q)$. 
	\item Else, $L<\frac17\Delta$, hence $\frac{\Delta}{2}\le\frac{\Delta}{1+\epsilon}-L\le\dfd(P,Q)\le\Delta+L\le2\Delta$. Recall that $i$ is chosen to be the maximal index such that $2^{i}\cdot L\le \frac\Delta2$. By maximality of $i$, we have $2\Delta\le2^{i+3}\cdot L$. Thus $\dfd(P,Q)\in\left[\frac\Delta2,2\Delta\right]\subseteq\left[2^i\cdot L,2^{i+3}\cdot L\right]$.
	As $\Delta<\frac{1}{\eps}\cdot L$ it follows that $i\le \log\frac{1}{2\eps}$. In particular we constructed an asymmetric bounded range distance oracle $\mathcal{O}_i$ for $P$, with the range $[2^{i-1}\cdot L,2^{i+3} \cdot L]$. 
	It follows that $\mathcal{O}_i$ returns a $(1+\eps)$-approximation for $\dfd(P,Q)$. 
\end{itemize}

To bound the space note that we constructed a single distance oracle using \Cref{thm:DOsymmetric} and $\log\frac1\eps$ distance oracles using \Cref{lem:DObounded-range} (each with constant ratio). Thus $O(\frac{1}{\eps})^{dk}\cdot\log\eps^{-1}$ in total.
Similarly, the expected preprocessing time is bounded by $m\log\frac{1}{\eps}\cdot\left(O(\frac{1}{\eps})^{kd}+O(d\log m)\right)$.
The query time is bounded by $O(kd\log(kd)\log d)$.

\DOasymmetric*

%% file: Streaming_Simplification.tex
\section{Curve simplification in the stream}\label{sec:streaming_simplification}
	Given a curve $P$ as a stream, our goal in this section is to maintain a $(k,1+\eps)$-simplification of $P$, with space bound that depend only on $k$ and $\eps$.
	A main ingredient is to maintain a $\delta$-simplification.	
	For the static model, Bereg \etal \cite{BJWYZ08} presented an algorithm that for fixed dimension that computed an optimal $\delta$-simplification in $O(m\log m)$ time. This algorithm was generalized to arbitrary dimension $d$ by the authors and Katz \cite{FFK20} (see \Cref{lem:optrsimplification}). 
	The algorithm is greedy: it finds the largest index $i$ such that $P[1,i]$ can be enclosed by a ball of radius $(1+\eps)\delta$, and then recurse for $P[i+1,m]$. The result is a sequence of balls, each of radius at most $(1+\eps)\delta$. The sequence of centers of the balls is a simplification of $P$ with distance $(1+\eps)\delta$. If the number of balls is larger than $k$, then $\dfd(P,\Pi)>\delta$ for every curve $\Pi$ with at most $k$ points. 
	The greedy algorithm essentially constructs a one-to-many paired walk along the resulted simplification and $P$.
	
	Let $\gMEB$ denote a streaming algorithm that maintains a $\gamma$-approximation of the minimum enclosing ball of a set of points. That is, in each point of time the algorithms has a center $\gMEB.c\in\R^d$ and a radius $\gMEB.r\in\reals$, such that all the points observed by this time are contained in $B_{\R^d}\left(\gMEB.c,\gMEB.r\right)$, and the minimum enclosing ball of the observed set of points has radius at least $\gMEB.r/\gamma$.
	In \Cref{subsec:MEB} we discuss several streaming \MEB~algorithms, all for $\gamma\in(1,2]$. For now, we will simply assume that we have such an algorithm as a black box.

	In \Cref{alg:GreedySimp} we describe a key sub-procedure of our algorithm called \texttt{GreedyStreamSimp}, which finds a greedy simplification while using $\gMEB$ as a black box. 
	\texttt{GreedyStreamSimp} receives as input a parameter $\delta$ and a curve $P$ in a streaming fashion, and returns a simplification $\Pi$ computed in a greedy manner.
	Specifically, it looks for the longest prefix of $P$ such that $\gMEB.r\le \delta$. That is the radius returned by $\gMEB$ is at most $\delta$
	Then, it continues recursively on the remaining points. Note that there is no bound on the size of the simplification $\Pi$ which  can have any size between $1$ and $|P|$.
	Nevertheless, we obtain a simplification $\Pi$ at distance at most $\delta$ from $P$, such that every curve at distance $\delta/\gamma$ from $P$ has length at least $\Pi$.
		
	\begin{algorithm}[h]
		\caption{\texttt{GreedyStreamSimp$(P,\delta)$}}\label{alg:GreedySimp}
		\DontPrintSemicolon
		\SetKwInOut{Input}{input}\SetKwInOut{Output}{output}
		\Input{A curve $P$, parameter $\delta>0$, a black box algorithm $\gMEB$}
		\Output{Simplification $\Pi$ of $P$, and a $\gMEB$ structure for the suffix of $P$ matched to the last point of $\Pi$.}
		\BlankLine
		Initialize $\gMEB$ with $\{P[1]\}$\;
		Set $\Pi\leftarrow \gMEB.c$\;
		\For{$i=2$ to $|P|$:}{
			Add $P[i]$ to $\gMEB$\;
			\If{$\gMEB.r\le\delta$}{
				Change the last point of $\Pi$ to $\gMEB(c)$\;
			}
			\Else{
				Initialize  $\gMEB$ with $P[i]$\;
				Set $\Pi\leftarrow\Pi\circ \gMEB.c$\;
			}
		}
		\Return $(\Pi,\gMEB)$\;
	\end{algorithm}

\begin{claim}\label{clm:greedy_iteration}
	Let $\Pi$ be a simplification computed by \emph{\texttt{GreedyStreamSimp}} with parameters $P,\delta$. Then $\dfd(P,\Pi)\le\delta$, and for any other simplification $\Pi'$ of $P$, if $\dfd(P,\Pi')\le\frac\delta\gamma$ then $|\Pi|\le |\Pi'|$.
\end{claim}
\begin{proof}
	First notice that $\dfd(P,\Pi)\le\delta$ is straightforward, because \texttt{GreedyStreamSimp} constructs a one-to-many paired walk $\omega$ along $\Pi$ and $P$ such that for each pair $(\Pi[i],\P_i)\in\omega$, $\P_i$ is contained in a ball of radius at most $\delta$.
	
	Let $\Pi'$ be a simplification of $P$ with $\dfd(P,\Pi')\le\frac\delta\gamma$, and consider an optimal walk $\omega$ along $\Pi'$ and $P$. If $\omega$ is not one-to-many, we remove vertices from $\Pi'$ until we get a simplification $\Pi''$ with  $\dfd(P,\Pi'')\le\frac\delta\gamma$ and an optimal one-to-many walk.
	Denote by $\Pi''_i$ the subcurve of $\Pi''$ that $\omega$ matches to $P[1,i]$, and by $A''_i$ the subsequence of points from $P[1,i]$ matched to the last point of $\Pi''_i$.
	In addition, denote by $\Pi_i$ and $\gMEB_i$ the state of these objects right after we finish processing $P[i]$, and let $A_i$ denote the subset of points from $P$ that where inserted to $\gMEB_i$.
	
	We show by induction on the iteration number, $i$, that either $|\Pi''_i|> |\Pi_i|$, or $|\Pi''_i|=|\Pi_i|$ and $|A''_i|\ge|A_i|$.
	For $i=1$ the claim is trivial. We assume that the claim is true for iteration $i\ge 1$, and prove that it also holds in iteration $i+1$.
	
	If in iteration $i$ we had $|\Pi''_i|> |\Pi_i|$, then in iteration $i+1$ either $|\Pi_{i+1}|=|\Pi_i|$ or $|\Pi_{i+1}|=|\Pi_i|+1$ and $|A_{i+1}|=1$, so the claim holds.
	
	Thus, assume that in iteration $i$ we had $|\Pi''_i|=|\Pi_i|$, and $|A''_i|\ge|A_i|$.
	If after adding $P[i+1]$ to $A_i$ we have $\gMEB.r>\delta$, then the minimum enclosing ball of $A_i\cup P[i+1]$ has radius larger than $\frac{\delta}{\gamma}$. This means that the minimum enclosing ball of $A''_i\cup P[i+1]$ also has radius larger than $\frac{\delta}{\gamma}$, and thus the length of both $\Pi''_i$ and $\Pi_i$ increase by 1, so $|\Pi''_{i+1}|=|\Pi_{i+1}|$, and $|A''_{i+1}|\ge|A_{i+1}|=1$.
	
	Else, if $\gMEB.r\le \delta$, then the minimum enclosing ball of $A_i\cup P[i+1]$ has radius at most $\delta$, and $A_{i+1}= A_i\cup P[i+1]$. If $|\Pi''_{i+1}|>|\Pi''_i|$ then we are done because $|\Pi''_i|=|\Pi_i|=|\Pi_{i+1}|$. Else, if $|\Pi''_{i+1}|=|\Pi''_i|$ then $P[i+1]$ is added to both $A''_i$ and $A_i$, and we get $|\Pi''_{i+1}|=|\Pi_{i+1}|$ and $|A''_{i+1}|\ge|A_{i+1}|$.
\end{proof}

\begin{algorithm}[h]
	\caption{\texttt{LeapingStreamSimp$(k,\gMEB,\init,\inc)$}}\label{alg:StreamSimp}
	\DontPrintSemicolon
	\SetKwInOut{Input}{input}\SetKwInOut{Output}{output}
	\Input{A curve $P$ in a streaming fashion, parameters $k\in\N$ and $\init\ge 1,\inc\ge 2$, a black box algorithm $\gMEB$}
	\Output{Simplification $\Pi$ of $P$ with at most $k$ points}
	\BlankLine
	Read $P[1,k+1]$ \tcp*{Ignore one of any two equal consecutive points}
	Set $\delta\leftarrow\init\cdot\frac12\min_{i\in[k]}\|P[i]-P[i+1]\|$\;
	Set $(\Pi,\gMEB)\leftarrow \texttt{GreedyStreamSimp}(P,\delta)$\;
	\For{$i\ge k+2$ to $m$:}{
		Read $P[i]$ and add it to $\gMEB$\;
		\If{$\gMEB.r\le\delta$}{
			Change the last point in $\Pi$ to $\gMEB(c)$\;
		}
		\Else{
			Initialize $\gMEB$ with $P[i]$\;
			Set $\Pi\leftarrow\Pi\circ \gMEB.c$\;
			\While{$|\Pi|=k+1$}{
				$\delta\leftarrow\delta\cdot\inc$\;
				Set $(\Pi,\gMEB)\leftarrow \texttt{GreedyStreamSimp}(\Pi,\delta)$\;
			}
		}
	}
	\Return $\Pi$\;
\end{algorithm}

	In \Cref{alg:StreamSimp} we present our main procedure for the streaming simplification algorithm called \texttt{LeapingStreamSimp}. Essentially, this algorithm tries to imitate the \texttt{GreedyStreamSimp} algorithm. Indeed, if we would know in advance the distance between $P$ to an optimal simplification $\Pi^*$ of length $k$, then we could find such a simplification by applying 
	\texttt{GreedyStreamSimp} with parameter $\gamma\cdot\delta^*$.
	However, as $\dfd(P,\Pi^*)$ is unknown in advance, \texttt{LeapingStreamSimp} tries to guess it.
	
	In addition to $k$ and $\gMEB$, \texttt{LeapingStreamSimp} also gets as input the parameters $\init\ge 1$ and $\inc\ge 2$. $\init$ is used for the initial guess of $\dfd(P,\Pi^*)$, while $\inc$ is used to update the current guess, when the previous guess is turned out to be too small.
	In more detail, \texttt{LeapingStreamSimp} starts by reading the first $k+1$ points (as up to this point our simplification is simply the observed curve). At this stage, the optimal simplification of length $k$ is at distance $\frac12\min_{i\in[k]}\|P[i]-P[i+1]\|$.
	The algorithm updates its current guess $\delta$ of $\dfd(P,\Pi^*)$ to $\init\cdot\frac12\min_{i\in[k]}\|P[i]-P[i+1]\|$, and execute \texttt{GreedyStreamSimp} on the $k+1$ observed points with parameter $\delta$. 
	Now the \texttt{LeapingStreamSimp} algorithm simply simulates \texttt{GreedyStreamSimp} with parameter $\delta$ as long as the simplification contains at most $k$ points. Once this condition is violated (that is, our guess turned out to be too small), the guess $\delta$ is multiplied by $\inc$. Now we compute a greedy simplification of the current simplification $\Pi$ using the new parameter $\delta$. This process is continued until we obtain a simplification of length at most $k$. At this point, we simply turn back to the previous simulation of the greedy simplification.

	As a result, eventually \texttt{LeapingStreamSimp} will hold an estimate $\delta$ and a simplification $\Pi$, such that $\Pi$ is an actual simplification constructed by the \texttt{GreedyStreamSimp} with parameter $\delta$. Alas, $\Pi$ was not constructed with respect to the observed curve $P$, but rather with respect to some other curve $P'$, where $\dfd(P,P')<\frac{2}{\inc}\delta$ (\Cref{clm:one_iteration}). Furthermore, the estimate $\delta$ will be bounded by the distance to the optimal simplification $\delta^*$, times $\approx\gamma\cdot\inc$ (\Cref{lem:streaming_dominates_offline}).

	In the analysis of the algorithm, by $\Pi_i$ and $\delta_i$ we refer to the state of the algorithm right after we finish processing $P[i]$.

	\begin{claim}\label{clm:one_iteration}
		After reading $m$ points, $\dfd(\Pi_m,P[1,m])\le(1+\frac{2}{\inc})\delta_m$. Moreover, there exists a curve $P'$ such that $\Pi_m$ is the simplification returned by \emph{\texttt{GreedyStreamSimp}} for the curve $P'$ and parameter $\delta_m$, and $\dfd(P',P[1,m])\le \frac{2\delta_m}{\inc}$.
	\end{claim}
	\begin{proof}
		The first part of the claim is a corollary that follows from the second part. Indeed, by \Cref{clm:greedy_iteration} we have $\dfd(\Pi_m,P')\le \delta_m$, and by the triangle inequality,
		\[
		\dfd(\Pi_m,P[1,m])\le \dfd(\Pi_m,P')+\dfd(P',P[1,m])\le (1+\frac{2}{\inc})\delta_m.
		\]
		
		We prove the second part by induction on $m$. For $m=k+1$ the claim is clearly true for $P'=P[1,k+1]$.
		Assume that the claim is true for $m-1\ge k+1$, so by the induction hypothesis there exists a curve $P'$ such that $\dfd(P',P[1,m-1])\le \frac{2\delta_{m-1}}{\inc}$, and $\Pi_{m-1}$ is the simplification returned by \texttt{GreedyStreamSimp} for the curve $P'$ and parameter $\delta_{m-1}$.
		
		If there is no leap step, then $\delta_m=\delta_{m-1}$. Let $P''=P'\circ P[m]$, then $\dfd(P'',P[1,m])\le \frac{2\delta_{m-1}}{\inc}=\frac{2\delta_m}{\inc}$.
		Since in this case \texttt{LeapingStreamSimp} imitates the steps of \texttt{GreedyStreamSimp}, we get that $\Pi_m$ will be exactly the simplification returned by \texttt{GreedyStreamSimp} for the curve $P''$ and parameter $\delta_m$.
		
		Else, if a leap step is taken, then $\delta_m=\delta_{m-1}\cdot \inc^h$ for some $h\ge 1$, so $\delta_{m-1}=\frac{\delta_m}{\inc^h}\le \frac{\delta_m}{\inc}$. 
		By the first part of the claim we have $\dfd(\Pi_{m-1},P[1,m-1])\le (1+\frac{2}{\inc})\delta_{m-1}$.
		
		Let $P''=\Pi_{m-1}\circ P[m]$, then for $\inc\ge 2$ \[
		\dfd(P'',P[1,m])\le\dfd(\Pi_{m-1},P[1,m-1])\le(1+\frac{2}{\inc})\delta_{m-1}\le(1+\frac{2}{\inc})\frac{\delta_m}{\inc}\le\frac{2\delta_m}{\inc} . \]
		The claim follows as the algorithms sets $\Pi_m$ to be the simplification returned by \texttt{GreedyStreamSimp} on the curve $P''$ and parameter $\delta_m$.
	\end{proof}

	Consider an optimal $k$-simplification $\Pi^*_m$ of $P[1,m]$, and denote $\delta^*_m=\dfd(\Pi^*_m,P[1,m])$.
	We will assume that $\inc>2\gamma$. 
	Set $\eta=\frac{\gamma\inc}{\inc-2\gamma}$.
	\begin{lemma}\label{lem:streaming_dominates_offline}
		Let $h$ be the minimal such that $\eta\cdot \delta_m^*\le\delta_{k+1}\cdot\inc^h$. Then $\delta_m\le \delta_{k+1}\cdot\inc^h$.
	\end{lemma}

	\begin{proof}
		Assume by contradiction that $\delta_m>\delta_{k+1}\cdot inc^h$, and let $i$ be the minimum index such that $\delta_i>\delta_{k+1}\cdot inc^h$. Then, when reading the $i$th point, the algorithm preforms a leap step (otherwise $\delta_i=\delta_{i-1}\le\delta_{k+1}\cdot inc^h$).
			
		By \Cref{clm:one_iteration}, there exists a curve $P'$ such that $\Pi_{i-1}$ is the simplification returned by \texttt{GreedyStreamSimp} for the curve $P'$ and parameter $\delta_{i-1}$, and $\dfd(P',P[1,i-1])\le \frac{2\delta_{i-1}}{inc}$. 
		
		Consider the time when the algorithm sets $\delta\leftarrow\delta_{k+1}\cdot inc^h$. Since $\delta_i>\delta_{k+1}\cdot inc^h$, the algorithm calls \texttt{GreedyStreamSimp} with the curve $P'\circ P[i]$ and parameter $\delta_{k+1}\cdot inc^h$, and get a simplification of length $k+1$ (otherwise, $\delta_i\le\delta_{k+1}\cdot inc^h$).
		
		Consider an optimal $k$-simplification $\tilde{\Pi}$ of the curve $P[1,i]$ with distance $\delta^*_m=\dfd(\tilde{\Pi},P[1,i])$
		By the triangle inequality, 
		\[
		\dfd(\tilde{\Pi},P'\circ P[i])\le \dfd(\tilde{\Pi},P[1,i])+\dfd(P[1,i],P'\circ P[i])\le \delta_m^*+\frac{2\delta_{i-1}}{inc}.
		\]
		 Therefore, by \Cref{clm:greedy_iteration}, the simplification returned by \texttt{GreedyStreamSimp} for the curve $P'\circ P[i]$ with parameter $\gamma(\delta_m^*+\frac{2\delta_{i-1}}{inc})$ has length at most $k$. 
		 But by the minimality of $i$ 
		\begin{align*}
		\gamma(\delta_{m}^{*}+\frac{2\delta_{i-1}}{inc}) & \le\frac{\gamma}{\eta}\delta_{k+1}\cdot inc^{h}+\frac{2\gamma}{inc}\delta_{k+1}\cdot inc^{h}\\
		& =\left(\frac{\inc-2\gamma}{\inc}+\frac{2\gamma}{inc}\right)\cdot\delta_{k+1}\cdot inc^{h}=\delta_{k+1}\cdot inc^{h}~.
		\end{align*}
		This contradicts the fact that \texttt{GreedyStreamSimp} returns a simplification of length $k+1$ when $\delta$ is set to $\delta_{k+1}\cdot  inc^h$.	
	\end{proof}

	We are now ready to prove the main theorem.
	\StreamSimplification*
	\begin{proof}
		The algorithm is very simple: for every $i\in [1,\lceil\log_{(1+\eps)}\frac{1}{\eps}\rceil]$, run the algorithm \texttt{LeapingStreamSimp}$(k,\gMEB,(1+\eps)^i,\frac1\eps)$, that is, \texttt{LeapingStreamSimp} with parameters $\init=(1+\eps)^i$ and $\inc=\frac1\eps$. 
		After observing the curve $P[1,m]$, the $i$'th instance of the algorithm will hold a simplification $\Pi_{m,i}$, and distance estimation $\delta_{m,i}$. The algorithm finds the index $i_{\min}$ for which $\delta_{m,i}$ is minimized, and returns $\Pi_{m,i}$ with $L=(1+\frac{2}{\inc})\delta_{m,i}$.
		
		Note that the space required for each copy of \texttt{LeapingStreamSimp} is $S(d,\gamma)+O(kd)$ as in each iteration it simply hold a single version of $\gMEB$ and at most $k$ points of the current simplification. Thus the space guarantee holds.  
		Recall that the optimal simplification is denoted by $\Pi_m^*$, where $\dfd(P[1,m],\Pi_m^*)=\delta_m^*$. 		
		We will argue that $\dfd(P[1,m],\Pi_{m,i_{\min}})=\gamma(1+O(\eps))\delta_m^*$. Afterwards, the $\eps$ parameter can be adjusted accordingly.
		
		We can assume that $m>k$, as otherwise we can simply return the observed curve $P$ (and $L=0$).
		Set $\delta_{\min}=\lambda[1,k+1]=\delta^*_{k+1}$\footnote{Recall that $\lambda[1,k+1]=\frac12\min_{i\in[k]}\|P[i]-P[i+1]\|$. }
		First note that as $\Pi_m^*$ contains at most $k$ points, 
		it follows that $\delta^*_m\ge \delta_{\min}$.
		Hence there are indices $1\le j\le \lceil\log_{(1+\eps)}\frac{1}{\eps}\rceil$ and $h\ge0$ such that
		\[
		(1+\eps)^{j-1}(\frac{1}{\eps})^{h}\delta_{\min}< \eta\cdot\delta_m^*\le (1+\eps)^j(\frac{1}{\eps})^{h}\delta_{\min}~.
		\]
		Note that for this particular $j$, we have that $(1+\eps)^j(\frac{1}{\eps})^{h}\delta_{\min}=\delta^j_{k+1}\cdot \inc^h$. 
		Hence by \Cref{lem:streaming_dominates_offline} we have that 
		$\delta_{m,j}\le (1+\eps)^j(\frac{1}{\eps})^{h}\delta_{\min}\le (1+\eps)\cdot \eta\cdot\delta_m^*$.
		Using \Cref{clm:one_iteration} we have that 		
		\begin{equation}
		\dfd(P[1,m],\Pi_{m,i_{\min}})\le(1+\frac{2}{\inc})\delta_{m,i_{\min}}\le(1+2\eps)\delta_{m}^{j}\le(1+4\eps)\cdot\eta\cdot\delta_{m}^{*}=(1+O(\eps))\cdot\gamma\cdot\delta_{m}^{*}~,\label{eq:simplificationMin}
		\end{equation}
		where the last step follows as
		$\eta=\frac{\gamma\inc}{\inc-2\gamma}=\gamma\cdot\frac{1}{1-2\gamma\eps}=(1+8\eps)\cdot\gamma$ for\footnote{For $\gamma=1+\eps$ and $\eps\le \frac14$ we will obtain $\eta\le (1+6\eps)\gamma$.} $\gamma\le 2$ and $\eps\le\frac{1}{8}$.
	\end{proof}

	\subsection{Approximating the minimum enclosing ball}\label{subsec:MEB}
	Computing the minimum enclosing ball of a set of points in Euclidean space is a fundamental problem in computational geometry.
	In the static setting, Megiddo \cite{Meg84} showed how to compute an $\MEB$ in $O(n\log n)$ time (for fixed dimension $d$), while Kumar \etal provided a static $(1+\eps)$-\MEB~algorithm running in $O(\frac{nd}{\eps}+\eps^{-4.5}\log\frac1\eps)$ time.
	
	A very simple $2\mbox{-}\MEB$ data structure in the streaming setting can be constructed as follows. Let $x$ be the first observed point, and set $2\mbox{-}\MEB.c\leftarrow x$ and $2\mbox{-}\MEB.r\leftarrow 0$. For each point $y$ in the remainder of the stream, set $2\mbox{-}\MEB.r\leftarrow \max\{2\mbox{-}\MEB.r, \Vert x-y\Vert \}$. In other words, this algorithm simply compute the distance from $x$ to its farthest point from the set. The approximation factor is $2$ because any ball that enclose $x$ and its farthest point $y$ has radius at least $\Vert x-y\Vert /2$. The space used by this algorithm is clearly $O(d)$. In addition, notice that the center of the ball in this algorithm is a point from the stream. This means that when using this $2\mbox{-}\MEB$ in our streaming simplification algorithm, we obtain a simplification $\Pi$ with points from $P$. This is sometimes a desirable property (for example, in applications from computational biology, see e.g. \cite{BJWYZ08}). Moreover, using \Cref{thm:StreamSimplification} we obtain a $(2+\eps)$ approximation factor, which is close to optimal because a vertex-restricted $k$-simplification is a $2$-approximation for an optimal (non-restricted) $k$-simplification of a given curve.
	\begin{corollary}
		For every parameters $\eps\in(0,\frac14)$ and $k\in\N$, there is a streaming algorithm which uses $O(\frac{\log\eps^{-1}}{\eps}\cdot kd))$ space, and given a curve $P$ in $\R^d$ in a streaming fashion, computes a vertex-restricted $(k,2+\eps)$-simplification $\Pi$ of $P$.
	\end{corollary}

	For a better approximation factor (using a simplification with arbitrary vertices), we can use the following $\gMEB$ algorithms. Chan and Pathak \cite{CP14} (improving over \cite{AS15}) constructed an $\gamma-\MEB$ algorithm for $\gamma=1.22$, also using $O(d)$ space. We conclude,
	\begin{restatable}{corollary}{StremingSimpEpsilon}\label{cor:StremingSimpEpsilon}
		For every parameters $\eps\in(0,\frac14)$ and $k\in\N$, there is a streaming algorithm which uses $O(\eps)^{-\frac{d+1}{2}}\log^{2}\eps^{-1}+O(kd\eps^{-1}\log\eps^{-1})$
		space, and given a curve $P$ in $\R^d$ in a streaming fashion, computes an $(k,1+\eps)$-simplification $\Pi$ of $P$.
	\end{restatable}

	Finally, as was observed by Chan and Pathak \cite{CP14}, using streaming techniques for $\eps$-kernels \cite{Zar11}, for every $\eps\in(0,\frac12)$ there is an $(1+\eps)\mbox{-}\MEB$ algorithm that using $O(\eps^{-\frac{d-1}{2}}\log\frac1\eps)$ space.
	\begin{lemma}\label{lem:1+epsMEB}
		For every parameter $\eps\in(0,\frac12)$, there is a $(1+\eps)\mbox{-}\MEB$ algorithm that uses $O(\eps)^{-\frac{d-1}{2}}\log\frac1\eps$ space.
	\end{lemma}
	As we relay heavily on \Cref{lem:1+epsMEB}, and do not aware of a published proof, we attach a proof sketch in \Cref{app:MEBproof}
	\begin{restatable}{corollary}{StremingSimpConstant}\label{cor:StremingSimpConstant}
		For every parameters $\eps\in(0,\frac14)$ and $k\in\N$, there is a streaming algorithm which uses $O(\frac{\log\eps^{-1}}{\eps}\cdot kd)$ space, and given a curve $P$ in $\R^d$ in a streaming fashion, computes an $(k,1.22+\eps)$-simplification $\Pi$ of $P$.
	\end{restatable}

%% file: Streaming_Distance_Oracle.tex
\section{Distance oracle: the streaming case}\label{sec:streaming_distance_oracle}
	Similarly to our static distance oracle, we first describe a construction (this time, in the streaming model) of a data structure that stores a $(k,r,\eps)$-cover $\C$ of size $O(\frac1\eps)^{kd}$ for $P$, and has a linear look-up time. Then, we show how to combine several of those structure (together with a streaming simplification) to produce a streaming distance oracle.

	\subsection{Cover of a curve}
	This entire subsection is dedicated to proving the following lemma.
	\begin{lemma}\label{lem:StramDOdecision}
		Given parameters $r\in \R_+$, $\eps\in(0,\frac14)$, 
		and $k\in \N$, there is a streaming algorithm that uses $O(\frac{1}{\eps})^{kd}$ space, and given a curve $P$ in $\R^d$ constructs a decision distance oracle with $O(kd)$ query time.
	\end{lemma}

	The cover that we describe below is an extended version of the cover described in section \Cref{subsec:StaticDecisionDO}, which also contains grid-curves of length smaller than $k$. Those smaller curves will allow us to update the cover when a new point of $P$ is discovered, when all we have is the new point and the previous cover.
	More precisely, we store a set of covers $\mathcal{C}_{i}=\{\mathcal{C}_{i,k'}\}_{1\le k'\le k}$ for $P[1,i]$, such that $\mathcal{C}_{{i},k'}$ is a $(k',\eps,r)$-cover for $P[1,i]$ with grid curves, exactly as we constructed for the static case. In other words, it contains exactly the set of all curves $W$ with $k'$ points from $\G_{i}=\bigcup_{1\le t\le i} G_{\eps,r}(P[t],(1+\eps)r)$ and $\dfd(P[1,i],W)\le(1+\eps)r$. We call such a cover a \emph{$(k',\eps,r)$-grid-cover}.
	
	\Cref{alg:ExtendCover} (\texttt{ExtendCover}), is a sub-routine that constructs the set of covers $\mathcal{C}_i$ for $P[1,i]$, given only the set $\mathcal{C}_{i-1}$ (for $i\ge 2$) and the new point $P[i]$. 
	\Cref{alg:GreedyDecision} (\texttt{StreamCover}) is the streaming algorithm that first reads $P[1]$ and construct a set $\C_1=\{\C_{1,k'}\}_{1\le k'\le k}$ such that $\C_{1,k'}$ is a $(k',\eps,r)$-grid-cover for $P[1]$,\footnote{\label{ftn:cover}This is the set of all the curves with at most $k$ points from $G_{\eps,r}(P[1],(1+\eps)r)$, with their distance to $P[1]$.} and then calls \texttt{ExtendCover} for each new observed point. 
		
	\begin{algorithm}[h]
		\caption{\texttt{StreamCover$(P,k,\eps,r)$}}\label{alg:GreedyDecision}
		\DontPrintSemicolon
		\SetKwInOut{Input}{input}\SetKwInOut{Output}{output}
		\Input{A curve $P$, parameters $r>0,k\in\N,\eps\in(0,\frac14)$}
		\Output{A $(k,\eps,r)$-cover $\C$ for $P$}
		\BlankLine
		Read $P[1]$.\;
		Construct a set $\C_1=\{\C_{1,k'}\}_{1\le k'\le k}$ such that $\C_{1,k'}$ is a $(k',\eps,r)$-grid-cover for $P[1]$.$^{\ref{ftn:cover}}$ \;
		Set $i\leftarrow 2$\;
		\While{read $P[i]$}{
			Set $\C_i\leftarrow$\texttt{ExtendCover$(\C_{i-1},P[i],k,\eps,r)$}\;
			Set $i\leftarrow i+1$ and delete $\C_{i-1}$\label{line:deleteOi-1}\;
		}
		Return $\C_i$\;
	\end{algorithm}

	\begin{algorithm}[h]
		\caption{\texttt{ExtendCover$(\C',p,k,\eps,r)$}}\label{alg:ExtendCover}
		\DontPrintSemicolon
		\SetKwInOut{Input}{input}\SetKwInOut{Output}{output}
		\Input{A $(k,\eps,r)$-cover $\C'$ for some (non-empty) curve $P'$, a point $p$, parameters $r>0,k\in\N,\eps\in(0,\frac14)$}
		\Output{A $(k,\eps,r)$-cover $\C$ for $P'\circ p$}
		\BlankLine
		
		Set $\C\leftarrow\emptyset$ and $\tilde{\C}\leftarrow\emptyset$\;
		Construct the set $\tilde{\C}$ of all the curves with at most $k$ points from $G_{\eps,r}(p,(1+\eps)r)$.\;	
		\For{each $W\in \C'$ of length $j<k$}{
			\For{each $X\in \tilde{\C}$ of length $k'\le k-j$}{
				Set $\dist\leftarrow\max\{\C'.\dist(W),\max_{1\le t\le k'}\{\Vert X[t]-p\Vert\}\}$\;
				\textbf{insert} $(W\circ X,\dist)$ into $\C$ 
				\tcp*{if $W\circ X\in\mathcal{C}$, keep the entry with minimal $\dist$}
			}
		}	
		\For{each $W\in \C'$ of length $k'$ \textbf{such that} $W[k']\in G_{\eps,r}(p,(1+\eps)r)$}{
			Set $\dist\leftarrow\max\{\C'.\dist(W),\|W[k']-p\|\}$\;
			\textbf{insert} $(W,\dist)$ into $\C$\tcp*{if $W\in\mathcal{C}$, keep the entry with minimal $\dist$}
		}
		Return $\C$\;
	\end{algorithm}

	Assume that $C'=\{\C_{i-1,j}\}_{1\le k'\le k}$ such that $\C_{i-1,k'}$ is a $(k',\eps,r)$-grid-cover for $P[1,i-1]$.
	We show that given the point $P[i]$ and $C'$, \Cref{alg:ExtendCover} outputs a set $\C=\{\C_{i,k'}\}_{1\le k'\le k}$ such that $\C_{i,k'}$ is a $(k',\eps,r)$-grid-cover for $P[1,i]$.
	
	Let $W$ be a curve with points from $\G_{i}=\bigcup_{1\le t\le i} G_{\eps,r}(P[t],(1+\eps)r)$ such that $\dfd(P[1,i],W)\le(1+\eps)r$.
	Consider an optimal walk along $W$ and $P[1,i]$, and let $j\le k'$ be the smallest index such that $P[i]$ is matched to $W[j]$. Notice that $W[j,k']$ is contained in $G_{\eps,r}(P[i],(1+\eps)r)$ and thus $W[j,k']$ is in $\tilde{C}$.

	\begin{figure}[h]
		\centering
		\includegraphics[scale=1.2]{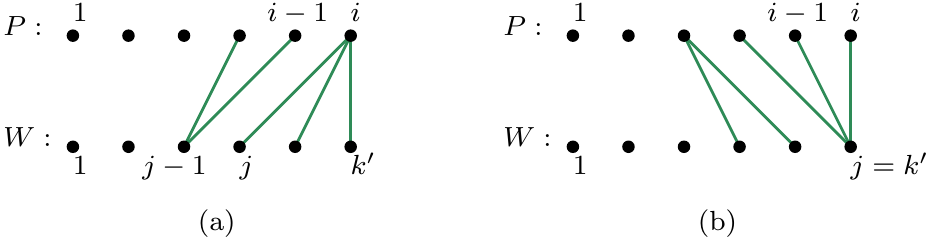}
		\caption{\label{fig:dynamic_prog}Constructing $C_{i,k'}$ from $\{C_{i-1,j}\}_{1\le j\le k}$.}
	\end{figure}

	If $P[i-1]$ is matched to $W[j-1]$ (see \Cref{fig:dynamic_prog}(a)) then
	$\dfd(P[1,i-1],W[1,j-1])\le (1+\eps)r$, and by the induction hypothesis $\mathcal{C}_{i-1,j-1}$ contains $W[1,j-1]$, with the value $\C'.\dist(W[1,j-1])=\dfd(P[1,i-1],W[1,j-1])$. Moreover, 
	$\dfd(P[1,i],W)=\max\{\dfd(P[1,i-1],W[1,j-1]), \dfd(P[i],W[j,k'])\}$, and indeed the algorithm inserts $W=W[1,j-1]\circ W[j,k']$ to $\C$ with the distance $\max\{\C'.\dist(W'[1,j-1]),\dfd(P[i],W[j,k'])\}$.
	
	Else, if $P[i-1]$ is matched to $W[j]$ (see \Cref{fig:dynamic_prog}(b)).
	Note that in this case it must be that $j=k'$ because $P[i]$ is also matched to $W[j]$.
	then $\dfd(P[1,i-1],W[1,k'])\le (1+\eps)r$, and by the induction hypothesis $\mathcal{C}_{i-1,k'}$ contains $W$, with the value $\C'.\dist(W)=\dfd(P[1,i-1],W)$. This time 
	$\dfd(P[1,i],W)=\max\{\dfd(P[1,i-1],W), \Vert W[k']-P[i]\Vert \}$, and indeed the algorithm inserts $W$ to $\C$ with the distance $\max\{\C'.\dist(W'),\Vert W[k']-P[i]\Vert \}$.
	
	The other direction (showing that if $W$ is in $\C$ then $W$ is a grid-curve with at most $k'$ points and $\dfd(P[1,i],W)\le(1+\eps)r$) can be proven by reversing the arguments.
	
	The space required for our algorithm is bounded by size of the set $\{\C_{i,k'}\}_{1\le k'\le k}$. Since each $\C_{i,k'}$ is exactly a $(k',\eps,r)$-grid-cover as described in \Cref{subsec:StaticDecisionDO}, we have $\left|\C_{i,k'}\right|=O(\frac{1}{\eps})^{k'd}$, and the total space is bounded by $\sum_{k'=1}^{k}\left|\C_{i,k'}\right|\le k\cdot O(\frac{1}{\eps})^{kd} = O(\frac{1}{\eps})^{kd}$.
	
	The query algorithm remains the same: Given a query $Q\in \R^k\times d$, we first compute a rounded curve $W'$ as in \Cref{subsec:StaticDecisionDO} (where $W'[i]$ is the closest grid point to $Q[i]$). If $W'$ is in the cover, we return $W'$ and $\dist(W')$, otherwise we return NO. The query time is thus $O(kd)$.
	
	\subsection{Cover with growing values of $r$}
	In the static scenario (\Cref{subs:static_asymmetric}), we used a set of bounded range distance oracles, each contains a set of covers (decision distance oracles) of the input curve $P$. However, we chose the ranges with respect to the distance $L$ between $P$ and a $(k,1+\eps)$-simplification $\Pi$ of $P$, and in the streaming scenario, this distance can increase in each round.
			
	The \texttt{StreamCover} algorithm presented in the previous subsection maintains a $(k,\eps,r)$-cover of $P$ for some given initial value $r$ which does not change. However, as more points of $P$ are read, it might be that the cover becomes empty (which also means that $L$ becomes larger than $r$). 
	Therefore, \Cref{alg:LeapingStreamingDecision} (\texttt{LeapingStreamCover}) presented in this subsection, simulates \texttt{StreamCover} until the cover becomes empty. Then, it increases $r$ by some given factor (similarly to the leap step in \Cref{alg:StreamSimp}), recompute the cover for the new value $r$, and continue simulating \texttt{StreamCover} for the new cover and $r$.
	Note that as in \Cref{alg:StreamSimp}, a leaping step can occur several times before the algorithm move on to the next point of $P$. Nevertheless, by first computing a simplification, we can actually compute how many leap steps are required without preforming them all.
		
	We start by reading the first $k+1$ points, and assume that there are no two consecutive identical points in $P[1,k+1]$ (otherwise ignore the duplicate, and continue reading until observing $k+1$ points without counting consecutive duplicates). Up until this point, we can simply compute a cover as we did in the static case. 

	Following the notation in the previous section, denote by $\Pi_m^*$ an optimal $k$-simplification of $P[1,m]$, and let $\delta^*_m=\dfd(P[1,m],\Pi_m^*)$.
	Denote by $r_m$ the value of $r$ at the end of round $m$ (i.e., when $P[m]$ is the last point read by \texttt{StreamCover}, right before reading $P[m+1]$).

	In addition to $\eps$ and $k$, the input for \Cref{alg:LeapingStreamingDecision} contains two parameters, $\init>0$ and $\inc\ge 2$. The parameter $\init$ is the initial value of $r$, and $\inc$ is the leaping factor by which we multiply $r$ when the cover becomes empty.
	
	\begin{algorithm}[h]
		\caption{\texttt{LeapingStreamCover$(P,k,\eps,\init,\inc)$}}\label{alg:LeapingStreamingDecision}
		\DontPrintSemicolon
		\SetKwInOut{Input}{input}\SetKwInOut{Output}{output}
		\Input{A curve $P$, parameters $k\in\N,\eps\in(0,\frac12)$, $\init>0,\inc\ge 2$}
		\Output{A $(k,\eps,r)$-cover $\C$ for $P$ for some $r\in \delta^{*}\cdot\left[\frac{1}{1+2\eps},(1+2\eps)\cdot\inc\right]$}
		\BlankLine
		Read $P[1,k+1]$.\;
		Set $r\leftarrow\init\cdot \lambda(P[1,k+1])$\;
		Construct the set $\C_{k+1}$ of $(k',\eps,r)$-grid-covers for $P[1,k+1]$, for $1\le k'\le k$.\;
		Set $i\leftarrow k+2$\;
		\While{read $P[i]$}{
			Set $\C_i\leftarrow$\texttt{ExtendCover$(\C_{i-1},P[i],k,\eps,r)$}\;
			\While{$\C_i=\emptyset$}{
				Let $W\in\C_{i-1}$ be an arbitrary curve\;
				Set $r\leftarrow r\cdot \inc$\;
				Set $\C_i\leftarrow$\texttt{StreamCover$(W\circ P[i],k,\eps,r)$}\;
			}
			Set $i\leftarrow i+1$ and delete $\C_{i-1}$\;
		}
		Return $\C=\C_i$\;
	\end{algorithm}

	Our goal is to maintain a set of covers with $r_m$ values that are not too far from $\delta^*_m$. For this,	we run our algorithm with $\inc=2^t$ for the minimum integer $t$ such that $2^t\ge\frac{25}{\eps}$, and $\init=2^i$ for some $i\in[0,t-1]$. Notice that $t=\log\frac1\eps+O(1)$.
	The intuition is that in order to get a good estimation for the true $\delta^*_m$, we will run $t$ instances of our algorithm, with initial $r$ values growing exponentially between $\delta^*_{k+1}=\lambda(P[1,k+1])$ and $\inc\cdot\delta^*_{k+1}$. Once an instance fail (i.e., its cover becomes empty), the $r$ value is multiplied by $\inc$ until the cover becomes non-empty. Roughly speaking, if $h$ is the number of times that we had to multiply the initial value $\init$ by $\inc$ so that the cover is non-empty, then $\inc^{h-1}\cdot\init\cdot \delta^*_{k+1}\lesssim\delta^*_m\lesssim \inc^h \cdot\init\cdot \delta^*_{k+1}$ (because otherwise $\Pi^*_m$ is an evidence that the cover is not empty after $h-1$ multiplications), and thus $\inc^h\cdot \init\cdot\delta^*_{k+1}\in\delta^*_m\cdot[\Theta(1),O(\frac{1}{\eps})]$.
	
	\begin{lemma}\label{lem:leapingStreaming}
		At the end of round $m$, $r_m\in\delta_{m}^{*}\cdot\left[\frac{1}{1+2\eps},(1+2\eps)\cdot\inc\right]$. Moreover, there exists a curve $P'$ such that $\C_m$ is a $(k,\eps,r_m)$-grid-cover for $P'$ and $\dfd(P',P[1,m])\le\frac{2}{\inc}\cdot r_m$.
	\end{lemma}
	\begin{proof}
		The proof is by induction on $m$. For the base case, $m=k+1$, note that after reading $P[1,k+1]$, we have $r_{k+1}=\init\cdot\delta^*_{k+1}\in \delta^*_{k+1}\cdot\left[1,\frac{\inc}{2}\right]\subseteq\delta_{m}^{*}\cdot\left[\frac{1}{1+2\eps},(1+2\eps)\cdot\inc\right]$, and $\C_{k+1}$ is a $(k,\eps,r_{k+1})$-grid-cover for $P'=P[1,k+1]$.
		
		For the induction step, suppose that $\C_{m-1}$ is a $(k,\eps,r_{m-1})$-grid-cover for a curve $P'$ where $r_{m-1}\in \delta_{m-1}^{*}\cdot \left[\frac{1}{1+2\eps},(1+2\eps)\cdot\inc\right]$ and $\dfd(P',P[1,m-1])\le\frac{2}{\inc}\cdot r_{m-1}$.
		
		If there is no leap step in round $m$, then $r_m=r_{m-1}$, and as shown in the previous subsection, \texttt{ExtendCover} returns a $(k,\eps,r_m)$-cover $\C_m$ for the curve $P'\circ P[m]$.
		We claim that the induction hypothesis holds w.r.t. $P'\circ P[m]$.
		Clearly, $\dfd(P'\circ P[m],P[1,m])\le\dfd(P',P[1,m-1])\le\frac{2}{\inc}\cdot r_{m-1}=\frac{2}{\inc}\cdot r_{m}$.
		Next, note that  $r_{m}=r_{m-1}\le(1+2\eps)\cdot\inc\cdot\delta_{m-1}^{*}\le(1+2\eps)\cdot\inc\cdot\delta_{m}^{*}$, because $\delta^*_{m-1}\le \delta^*_m$. Finally, as there was no leap step, there is some curve $W\in \C_m$ such that $\dfd(W,P'\circ P[m])\le(1+\eps)r_m$. By the triangle inequality,
		\begin{align*}
		\delta_{m}^{*}\le\dfd(W,P[1,m]) & \le\dfd(W,P'\circ P[m])+\dfd(P'\circ P[m],P[1,m])\\
		& \le(1+\eps)r_{m}+\frac{2}{\inc}\cdot r_{m}\le(1+2\eps)r_{m}~.
		\end{align*} 
		We conclude that  	$r_m\in\delta_{m}^{*}\cdot\left[\frac{1}{1+2\eps},(1+2\eps)\cdot\inc\right]$. 
		
		Next, we consider the case where round $m$ is a leap step.
		As we preformed a leap step, $\C_m=\emptyset$, so there is no grid-curve at distance at most $(1+\eps)r_{m-1}$ from $P'\circ P[m]$. Therefore, it follows from the triangle inequality
		that there is no curve at distance $r_{m-1}$ from $P'\circ P[m]$,
		as by rounding any curve we get a grid-curve within distance $\eps r_{m-1}$ from it. In particular, 
		\begin{align}
		\delta_{m}^{*}=\dfd(\Pi_{m}^{*},P[1,m]) & \ge\dfd(\Pi_{m}^{*},P'\circ P[m])-\dfd(P'\circ P[m],P[1,m])\nonumber \\
		& \ge r_{m-1}-\frac{2}{\inc}r_{m-1}=(1-\frac{2}{\inc})\cdot r_{m-1}~.\label{eq:delta*greaterrm-1}
		\end{align}
		
		Let $W_{m-1}$ be an arbitrary grid-curve at distance $(1+\eps)r_{m-1}$ from $P'$ chosen by the algorithm. The algorithm choose $r_m=r_{m-1}\cdot\inc^{h}$, such that $h\ge1$ is the minimal integer such that there is a grid-curve $W_m$ of length $k$ at distance $(1+\eps)r_m$ from $W_{m-1}\circ P[m]$. 
		The algorithm then constructs a $(k,\eps,r_m)$-grid-cover for $W_{m-1}\circ P[m]$. It holds that
		\begin{align}
		\dfd(W_{m-1}\circ P[m],P[1,m]) & \le\dfd(W_{m-1},P')+\dfd(P',P[1,m-1])\nonumber \\
		& \le(1+\eps)r_{m-1}+\frac{2}{\inc}\cdot r_{m-1}\le2\cdot r_{m-1}\le\frac{2}{\inc}r_{m}~.\label{eq:Wmbound}
		\end{align}
		It remains to prove that $r_m$ is in $\delta_{m}^{*}\cdot\left[\frac{1}{1+2\eps},(1+2\eps)\cdot\inc\right]$. Firstly, 
		\begin{align*}
		(1+\eps)r_{m} & \ge\dfd(W_{m},W_{m-1}\circ P[m])\\
		& \ge\dfd(W_{m},P[1,m])-\dfd(P[1,m],P'\circ P[m])-\dfd(P'\circ P[m],W_{m-1}\circ P[m])\\
		& \ge\delta_{m}^{*}-\frac{2}{\inc}\cdot r_{m-1}-(1+\eps)\cdot r_{m-1}~,
		\end{align*}
		where the third inequality holds as every length $k$ curve is at distance at least $\delta_m^*$ from $P[1,m]$, and the induction hypothesis. It follows that $\delta_{m}^{*}\le\left(1+\eps+\frac{3+\eps}{\inc}\right)\cdot r_{m}\le(1+2\eps)\cdot r_{m}$.
		
		For the second bound we continue by case analysis.
		If $h=1$, then $r_m=\inc\cdot r_{m-1}$, and by \cref{eq:delta*greaterrm-1}, $\delta_{m}^{*}\ge(1-\frac{2}{\inc})\cdot r_{m-1}\ge\frac{1-\eps}{\inc}\cdot r_{m}$. Thus $r_m\le (1+2\eps)\cdot\inc\cdot \delta_{m}^{*}$.
		Else, $h\ge2$ thus $r_{m-1}\cdot \inc^2\le r_m$. It follows that there is no grid curve of length $k$ at distance $(1+\eps)\cdot \frac{r_m}{\inc}$ from $W_{m-1}\circ P[m]$. In particular, there is no length $k$ curve at distance  $\frac{r_m}{\inc}$ from $W_{m-1}\circ P[m]$. Hence 
		\begin{align*}
		\delta_{m}^{*}=\dfd(\Pi_{m}^{*},P[1,m]) & \ge\dfd(\Pi_{m}^{*},W_{m-1}\circ P[m])-\dfd(W_{m-1}\circ P[m],P[1,m])\\
		& \overset{(\ref{eq:Wmbound})}{\ge}\frac{r_{m}}{\inc}-2\cdot r_{m-1}\ge\frac{r_{m}}{\inc}-\frac{2}{\inc^{2}}\cdot r_{m}\ge(1-\eps)\cdot\frac{r_{m}}{\inc}~,
		\end{align*}
		Implying $r_m\le (1+2\eps)\cdot\inc\cdot \delta_{m}^{*}$.
		The lemma now follows.
	\end{proof}

	\subsection{General distance oracle}
	The high level approach that we use here is similar to \Cref{subs:static_asymmetric}, except that the simplification and the oracles have to be computed in a streaming fashion. The main challenge is therefore that the value $L\approx\dfd(P,\Pi)$ on which the entire construction of \Cref{subs:static_asymmetric} relay upon, is unknown in advance.
	 
	The objects stored by our streaming algorithm in round $m$ are as follows: an approximated $k$-simplification $\Pi_m$ of the observed input curve $P[1,m]$, with a value $L_m$ (from \Cref{thm:StreamSimplification}), a distance oracle $\mathcal{O}_{\Pi_m}$ for $\Pi_m$ (from \Cref{thm:DOsymmetric}), and a set of $O(\log\frac{1}{\eps})$ covers of $P[1,m]$ computed by the \texttt{LeapingStreamCover} algorithm.
	
	For the sake of simplicity, we assume that there are no two equal consecutive points among the first $k+1$ points in the data stream (as we can just ignore such redundant points).
	
	\paragraph{The algorithm.}
	First, using \Cref{cor:StremingSimpEpsilon} we maintain a $(k,1+\eps)$-simplification $\Pi_m$ of $P[1,m]$ with a value $L_m$ such that 
	\begin{equation}\label{eq:L_m}
	\delta_m^*\le \dfd(P[1,m],\Pi_m)\le L_m\le (1+\eps)\delta_m^*\le(1+\eps)\dfd(P[1,m],\Pi_m)~,
	\end{equation}
	where the first and last inequalities follow as $\delta_m^*$ is the minimal distance from $P[1,m]$ to any curve of length $k$.
	In addition, at the end of each round, using \Cref{thm:DOsymmetric}, we will compute a static $(1+\eps)$-distance oracle $\mathcal{O}_{\Pi_m}$ for $\Pi_m$.
		
	\sloppy Secondly, let $t$ be the minimum integer such that $2^t\ge\frac{25}{\eps}$.
	As in the previous subsection, we set $\inc=2^t$ and $\init_i=2^i$ for $i\in[0,t-1]$. Then we run $t$ instances of \texttt{LeapingStreamingDecision} simultaneously: for every $i\in[0,t-1]$, we run  \texttt{LeapingStreamingDecision$(P,k,\eps,\init_i,\inc)$}.
	
	Denote by $\C_{i,m}$ the $(k,\eps,r_{i,m})$-cover created by the execution of \texttt{LeapingStreamingDecision$(P,k,\eps,\init_i,\inc)$} at the end of round $m$, where $r_{i,m}$ is the distance parameter of the cover $\C_{i,m}$. Note that $r_{i,m}=2^i\cdot \inc^j\cdot\delta^*_{k+1}$ for some index $j\ge0$.
	By \Cref{obs:cover} and \Cref{lem:leapingStreaming}, at the end of round $m$ we have a $(k,2\eps,r_{i,m})$-decision distance oracle $\mathcal{O}_{i,m}$ for some curve $P'$ such that $\dfd(P[1,m],P')\le\frac{2}{\inc}r_{i,m}$.
	By \Cref{lem:leapingStreaming},
	\begin{equation}
	r_{i,m}=2^i\cdot \inc^j\cdot\delta^*_{k+1}\in\left[\frac{1}{1+2\eps},(1+2\eps)\cdot\inc\right]\cdot\delta_{m}^{*}\overset{(\ref{eq:L_m})}{\subseteq}\left[\frac{1}{1+4\eps},(1+2\eps)\cdot\inc\right]\cdot L_{m}~.\label{eq:r_m_iBoundary}
	\end{equation}
	The query algorithm follows the lines in \Cref{subs:static_asymmetric}.  
	Given a query curve $Q\in\reals^{d\times m}$, we query $\mathcal{O}_{\Pi_m}$ and get a value $\Delta$ such that $\dfd(Q,\Pi_m)\le\Delta\le(1+\eps)\dfd(Q,\Pi_m)$.
	\begin{itemize}
		\item If $\Delta\ge\frac{1}{\eps}\cdot L_m$, return $(1+\eps)\Delta$.
		\item Else, if $\Delta\le 3 L_m$, let
		$j\ge0$ and $i\in[0,t-1]$ be the unique indices such that 
		${2^i\cdot \inc^j\cdot\delta^*_{k+1}\le 10\cdot L_m <  2^{i+1}\cdot \inc^j\cdot\delta^*_{k+1}}$.\footnote{\label{ftn:unique_indexes}Such indexes $i,j$ exists as $L_m\ge\delta^*_m\ge\delta^*_{k+1}$, and they are unique because $inc=2^t$.} Return $\O_{i,m}(Q)+\frac{2}{\inc}r_{i,m}$.
		\item Else, set $\alpha=\lceil\frac{\Delta}{L_m}\rceil\in[4,\lceil\frac1\eps\rceil]$, and let
		$j\ge0$ and $i\in[0,t-1]$ be the unique indices such that 
		${2^{i}\cdot \inc^j\cdot\delta^*_{k+1}\le 10\cdot\alpha L_m < 2^{i+1}\cdot \inc^j\cdot\delta^*_{k+1}}$.$^{\ref{ftn:unique_indexes}}$ Return $\O_{i,m}(Q)+\frac{2}{\inc}r_{i,m}$.
	\end{itemize}
	
	\paragraph{Correctness.} We show that in each of the above cases, the query algorithm returns a value in ${\left[1,1+O(\eps)\right]\cdot \dfd(P,Q)}$. Afterwards, the $\eps$ parameter can be adjusted accordingly. By the triangle inequality it holds that
	\begin{equation}\label{eq:triangle1}
	\dfd(P[1,m],Q) \le\dfd(Q,\Pi_{m})+\dfd(P[1,m],\Pi_{m})\le\Delta+L_{m}~,
	\end{equation}
	and
	\begin{equation}\label{eq:triangle2}
	\dfd(P[1,m],Q)\ge\dfd(Q,\Pi_{m})-\dfd(P[1,m],\Pi_{m})\ge\frac{\Delta}{1+\epsilon}-L_{m}~.
	\end{equation}
	\begin{itemize}
		\item If $\Delta\ge\frac{1}{\eps}\cdot L_m$, then $\dfd(P[1,m],Q) \overset{(\ref{eq:triangle1})}{\le}(1+\eps)\Delta$, and $\dfd(P[1,m],Q)\overset{(\ref{eq:triangle2})}{\ge}\frac{1}{1+\eps}\Delta-\eps\Delta\ge(1-2\eps)\Delta$. It follows that $(1+\eps)\Delta\le\frac{1+\eps}{1-2\eps}\cdot\dfd(P[1,m],Q)=(1+O(\eps))\cdot\dfd(P[1,m],Q)$.
	\end{itemize}
		
	For the next two cases, we first show that there exists a value $\phi$ such that $\dfd(P[1,m],Q)\in [\frac{1}{4},4]\cdot \phi$ (each case has a different $\phi$ value). Then, the rest of the analysis for both cases continues at $\clubsuit$ (as it is identical given $\phi$).
	\begin{itemize}
		\item If $\Delta\le 3 L_m$ then $\dfd(P[1,m],Q)\overset{(\ref{eq:triangle1})}{\le} 4 L_m$.
		We also have that $\dfd(P[1,m],Q)\ge\dfd(P[1,m],\Pi_{m}^{*})=\delta^*_m\overset{(\ref{eq:L_m})}{\ge}\frac{L_{m}}{1+\eps}$. Hence $\dfd(P[1,m],Q)\in [\frac{1}{1+\eps},4]\cdot L_m\subseteq [\frac{1}{4},4]\cdot L_m$.
		Set $\phi=L_m$.
		
		\item Else, $L_m<\frac{1}{3}\Delta$, and hence
		$\dfd(P[1,m],Q)\overset{(\ref{eq:triangle1})}{\le}\frac43\Delta$ and
		$\dfd(P[1,m],Q)\overset{(\ref{eq:triangle2})}{\ge}\frac{1}{1+\eps}\Delta-\frac13\Delta>\frac{1}{3}\Delta$. 
		By the definition of $\alpha$, it holds that $\Delta\le\alpha\cdot L_{m}\le(\frac{\Delta}{L_{m}}+1)\cdot L_{m}\le\frac{4}{3}\Delta$, hence $\Delta\in[\frac34,1]\cdot \alpha L_m$. Thus
		$
		\dfd(P[1,m],Q)\in[\frac{1}{3}\Delta,\frac43\Delta]\subseteq[\frac{1}{4},\frac43]\cdot \alpha L_{m}\subseteq[\frac{1}{4},4]\cdot \alpha L_{m}
		$.
		Set $\phi=\alpha L_{m}$.
	\end{itemize}			
	{\huge $\clubsuit$}
	We have $\dfd(P[1,m],Q)\in[\frac{1}{4},4]\cdot \phi$, and let $j\ge0$ and $i\in[0,s-1]$ be the unique indices such that $2^i\cdot \inc^j\cdot\delta^*_{k+1}\le 10\cdot \phi <  2^{i+1}\cdot \inc^j\cdot\delta^*_{k+1}$.$^{\ref{ftn:unique_indexes}}$
	Consider the decision distance oracle $\mathcal{O}_{i,m}$.
	If $r_{i,m}\neq 2^{i}\cdot \inc^j\cdot\delta^*_{k+1}$, then one of the followings hold in contradiction to \cref{eq:r_m_iBoundary}:
	\begin{align*}
	r_{m}^{i} & \ge2^{i}\cdot\inc^{j+1}\cdot\delta^*_{k+1}\ge\frac{\inc}{2}\cdot2^{i+1}\cdot\inc^{j}\cdot\delta^*_{k+1}\ge\frac{\inc}{2}\cdot10\cdot\phi>(1+2\eps)\cdot\inc\cdot L_{m}~,\\
	r_{m}^{i} & \le2^{i}\cdot\inc^{j-1}\cdot\delta^*_{k+1}\le\frac{1}{\inc}\cdot2^{i}\cdot\inc^{j}\cdot\delta^*_{k+1}\le\frac{1}{\inc}\cdot10\cdot\phi\le\frac{\alpha}{\inc}\cdot10\cdot L_{m}\overset{(*)}{<}\frac{1}{1+4\eps}\cdot L_{m}~,
	\end{align*}
	where the inequality $^{(*)}$ follows as $\frac{\alpha}{\inc}\cdot10\le\lceil\frac{1}{\eps}\rceil\cdot\frac{\eps}{25}\cdot10<\frac{1}{2}<\frac{1}{1+4\eps}$.
	We conclude that $r_{i,m}= 2^{i}\cdot \inc^j\cdot\delta^*_{k+1}$ and hence $r_{i,m}\le 10\cdot \phi<2\cdot r_{i,m}$.
			
	Following \Cref{lem:leapingStreaming}, let $P'$ be the curve for which $\C_{i,m}$ is an $(k,\eps,r_{i,m})$ cover. Then 
	\begin{align*}
	\dfd(P',Q) & \le\dfd(P',P[1,m])+\dfd(P[1,m],Q)\\
	& \le\frac{2}{\inc}\cdot r_{i,m}+4\phi\le\left(\frac{2}{\inc}+\frac{4}{5}\right)r_{i,m}\le r_{i,m}
	\end{align*}
	Thus $\O_{i,m}$ will return some value. Recall that our algorithm returns $\O_{i,m}(Q)+\frac{2}{\inc}r_{i,m}$. 
	It holds that 
	\[
	\O_{i,m}(Q)\ge\dfd(P',Q)\ge\dfd(P[1,m],Q)-\dfd(P',P[1,m])\ge\dfd(P[1,m],Q)-\frac{2}{\inc}r_{i,m}~.
	\]
	In addition,
	\begin{align*}
	\O_{i,m}(Q)\le(1+\eps)\dfd(P',Q)\le\dfd(P',Q)+\eps r_{i,m} & \le\dfd(P[1,m],Q)+\dfd(P',P[1,m])+\eps r_{i,m}\\
	& \le\dfd(P[1,m],Q)+\frac{2}{\inc}r_{i,m}+\eps r_{i,m}~,
	\end{align*}
	It follows that the returned value is bounded by $\dfd(P[1,m],Q)+\left(\frac{4}{\inc}+\eps\right)r_{i,m}\le\dfd(P[1,m],Q)+\left(\frac{4}{\inc}+\eps\right)10\cdot \phi=(1+O(\eps))\cdot\dfd(P[1,m],Q)$.

	\paragraph{Space and Query Time} To maintain a simplification $\Pi_m$, according to \Cref{cor:StremingSimpEpsilon} we used $O(\eps)^{-\frac{d+1}{2}}\log^{2}\eps^{-1}+O(kd\eps^{-1}\log\eps^{-1})$ space.
	The distance oracle $\mathcal{O}_{\Pi_m}$, require $O(\frac{1}{\eps})^{dk}\cdot\log\eps^{-1}$ space by \Cref{thm:DOsymmetric}. Finally, we use $O(\log \frac1\eps)$ decision distance oracles, with covers constructed by the \texttt{LeapingStreamCover} algorithm. As this algorithm simulates \texttt{StreamCover}, by \Cref{lem:StramDOdecision} the space consumption is $O(\frac{1}{\eps})^{kd}$.
	We conclude that the total space used by our streaming algorithm is $O(\frac{1}{\eps})^{dk}\cdot\log\eps^{-1}$.
	
	Regarding query time, we first query $\Pi_m$, which takes $\tilde{O}(kd)$ time (\Cref{cor:StremingSimpEpsilon}). Afterwards, we might preform another query in $O(kd)$ time (\Cref{lem:StramDOdecision}). All other computations take $O(kd)$ time. The theorem follows.
	
	\mainStreamingDO*

%% file: Subcurve_Distance_Oracle.tex
	\section{Distance oracle to a sub-curve and the ``Zoom-in'' problem}\label{sec:subcurve_distance_oracle}
	In this section we consider the following generalization of the distance oracle problem.

	\begin{problem}
		Given a curve $P\in \reals^{d\times m}$ and parameter $\eps>0$, preprocess $P$ into a data structure that given a query curve $Q\in \reals^{d\times k}$, and two indexes $1\le i\le j\le m$, returns an $(1+\eps)$-approximation of $\dfd(P[i,j],Q)$.
	\end{problem}
	
	A trivial solution is to store for any $1\le i\le j\le m$ a distance oracle for $P[i,j]$, then the storage space increases by a factor of $m^2$. In cases where $m$ is large, one might wish to reduce the storage space at the cost of increasing the query time or approximation factor.
	
	We begin by introducing a new problem called the ``zoom-in'' problem, which is closely related to the above problem. Our solution to the ``zoom-in'' problem will be used as a skeleton for a distance oracle to a sub-curve.
	
	\subsection{The ``zoom-in'' problem}
	When one needs to visualize a large curve, it is sometimes impossible to display all its details, and displaying a simplified curve is a natural solution. In some visualization applications, the user wants to ``zoom-in'' and see a part of the curve with the same level of details. For example, if the curve represents the historical prices of a stock, one might wish to examine the rates during a specific period of time. In such cases,  a new simplification needs to be calculated. In the following problem, we wish to construct a data structure that allows a quick zoom-in (or zoom-out) operation.
	
	\begin{problem}[Zoom-in to a curve]
		Given a curve $P\in \reals^{d\times m}$ and an integer $1\le k< m$, preprocess $P$ into a data structure that given $1\le i< j\le m$, return an optimal $k$-simplification of $P[i,j]$.
	\end{problem}

	To make the space and preprocessing time reasonable, we introduce a bi-criteria approximation version of the zoom-in problem:
	Instead of returning an optimal $k$-simplification of $P[i,j]$, the data structure will return an $(\alpha,k,\gamma)$-simplification of $P[i,j]$ (i.e. a curve $\Pi\in\reals^{d\times \alpha k}$ such that $\dfd(P[i,j],\Pi)\le \gamma\dfd(P[i,j],\Pi')$ for any $\Pi'\in \reals^{d\times k}$).
	
	We will use the two following observations.	
	\begin{observation}\label{obs:concatination}
		Let $\{P_i\}_{i=1}^{s}$, $\{Q_i\}_{i=1}^{s}$ be curves. Then $\dfd(P_1\circ P_2\circ\dots\circ P_s,Q_1\circ Q_2\circ\dots\circ Q_s)\le \max_i\{\dfd(P_i,Q_i)\}$.
	\end{observation}
	
	\begin{observation}\label{obs:subcurve_simp}
		Let $P$ be a curve and $P'$ a sub-curve of $P$. Let $\Pi'$ be a $(k,\gamma)$-simplifications of $P'$. Then for any $\Pi\in\R^{d\times k}$ it holds that $\dfd(P',\Pi')\le\gamma\dfd(P,\Pi)$.
	\end{observation}
	\begin{proof}
		Consider a paired walk $\omega$ along $P$ and $\Pi$, and let $\Pi''$ be a sub-curve of $\Pi$ that contains all the points of $\Pi$ that were matched by $\omega$ to the points of $P'$.
		Then clearly $\dfd(P',\Pi'')\le \dfd(P,\Pi)$, and by the definition of $(k,\gamma)$-simplification, $\dfd(P',\Pi')\le \gamma \dfd(P',\Pi'')\le \gamma\dfd(P,\Pi)$.
	\end{proof}
	
	Below, we present a data structure for the zooming problem with $O(mk\log\frac{m}{k})$ space, which returns $(k,1+\eps,2)$-simplifications in $O(kd)$ time.
	
	For simplicity of the presentation, we will assume that $m$ is a power of $2$ (otherwise, add to the curve $P$, $2^{\lceil\log m\rceil}-m$ copies of the point $P[m]$).
	Construct a recursive structure with $\log\frac{m}{k}$ levels as follows. The first level contains a $(k,1+\eps,1)$-simplifications of $P[i,\frac{m}{2}]$ for any $1\le i\le \frac{m}{2}$,
	and $P[\frac{m}{2}+1,j]$ for any $\frac{m}{2}+1\le j\le m$. In the second level, we recurs with $P[1,\frac{m}{2}]$ and $P[\frac{m}{2}+1,m]$. The $i$th level corresponds to $2^i$ sets of simplifications, each set corresponds to a sub-curve of length $\frac{m}{2^i}$. In the last level, the length of the corresponding sub-curves is at most $k$.
	The total space of the data structure is $O(mkd\log\frac{m}{k})$, this is as each point $P[i]$ is responsible for a single simplification (a curve in $\R^{d\times k}$) in $\log\frac{m}{k}$ different levels.
	As all the curves at the $i$ level of the recursion have length $\frac{m}{2^i}$, using \Cref{thm:SimplificationHighDim}, the preprocessing time is  
	\[
	\sum_{i=1}^{\log\frac{m}{k}}m\cdot\tilde{O}(\frac{m}{2^{i}}\cdot\frac{d}{\eps^{4.5}})=\tilde{O}(m^{2}d\eps^{-4.5})~.
	\]
	If $d$ is fixed, then according to \Cref{thm:SimplificationHighDim}, the preprocessing time will be 
	\[
	\sum_{i=1}^{\log\frac{m}{k}}m\cdot O\left(\frac{m}{2^{i}}\cdot(\frac{1}{\eps}+\log\frac{m}{2^{i}\eps}\log\frac{m}{2^{i}})\right)=O\left(m^{2}\cdot(\frac{1}{\eps}+\log\frac{m}{\eps}\log m)\right)=\tilde{O}(\frac{m^{2}}{\eps})~.
	\]
	
	Given two indexes $1\le i< j\le m$, if $j-i\le k$, simply return $P[i,j]$. Else, let $t$ be the smallest integer such that $i\le x\cdot \frac{m}{2^t}< j$ for some $x\in [2^{t-1}]$. 
	Let $\Pi_1$ and $\Pi_2$ be the simplifications of $P[i, x\cdot \frac{m}{2^t}]$ and $P[ x\cdot \frac{m}{2^t}+1,j]$, respectively. Return the concatenation $\Pi_1\circ\Pi_2$.
	
	We argue that $\Pi_1\circ\Pi_2$ is a $(k,1+\eps,2)$-simplification of $P[i,j]$. Indeed, let $\Pi\in\R^{d\times k}$ be an arbitrary length $k$ curve.  By \Cref{obs:subcurve_simp}, we have $\dfd(P[i, x\cdot \frac{m}{2^t}],\Pi_1)\le(1+\eps)\dfd(P[i,j],\Pi)$ and $\dfd(P[x\cdot\frac{m}{2^t}+1,j],\Pi_2)\le(1+\eps)\dfd(P[i,j],\Pi)$. By \Cref{obs:concatination} we conclude that $\dfd(P[i,j],\Pi_1\circ\Pi_2)\le(1+\eps)\dfd(P[i,j],\Pi)$.
	
	\zoomin*

	\subsection{$(1+\eps)$-factor distance oracle to a sub-curve}
	Notice that as described in \Cref{obs:trivial_oracle}, by the triangle inequality, a solution to the zooming problem can be used in order to answer approximate distance queries to a sub-curve in $O(k^2d)$ time. However, the approximation factor will be constant.	
	
	Our simplification for $P[i,j]$ is obtained by finding a partition of $P[i,j]$ into two disjoint sub-curves, for which we precomputed an $(k,1+\eps)$-simplifications. To achieve a $(1+\eps)$ approximation factor, instead of storing $(k,1+\eps)$-simplifications, we will store distance oracles that will be associated with the same set of sub-curves, and then find an optimal matching between the query $Q$ and a partition of $P[i,j]$.
	
	Let $\mathcal{O}$ be a $(1+\eps)$-distance oracle with storage space $S(m,k,d)$, query time $T(m,k,d)$, and $PT(m,k,d)$ expected preprocessing time. Using \Cref{thm:DOasymmetric} we can obtain $S(m,k,d)=O(\frac{1}{\eps})^{dk}\cdot\log\eps^{-1}$, $T(m,k,d)=\tilde{O}(kd)$, and  $PT(m,k,d)=m\log\frac{1}{\eps}\cdot\left(O(\frac{1}{\eps})^{kd}+O(d\log m)\right)$.

	Given two indexes $1\le i< j\le m$, if $j-i\le k$, simply compute and return $\dfd(P[i,j],Q)$ in $O(k^2d)$ time. 
	Else, let $t$ and $x$ be the integers as in the previous subsection, set $y=x\cdot \frac{m}{2^t}$, and let $\mathcal{O}_1$ and $\mathcal{O}_2$ be distance oracles for $P[i,y]$ and $P[y+1,j]$, respectively. Return 
	\begin{align*}
	\tilde{\Delta}=\min \{ 
	&\min_{1\le q\le k} \max\{\mathcal{O}_1(Q[1,q]),\mathcal{O}_2(Q[q,k])\},\\
	&\min_{1\le q\le k-1} \max\{\mathcal{O}_1(Q[1,q]),\mathcal{O}_2(Q[q+1,k])\} \}.
	\end{align*}
	The query time is therefore $O(k^2d + k\cdot T(m,k,d))=\tilde{O}(k^2d)$. The storage space is
	$m\log\frac{m}{k}\cdot S(m,k,d)=m\log m\cdot O(\frac{1}{\eps})^{dk}\cdot\log\frac1\eps$ because we construct $m\log\frac{m}{k}$ distance oracles (instead of simplifications).
	The expected preprocessing time is 
	\begin{align*}
	\sum_{i=1}^{\log\frac{m}{k}}m\cdot PT(\frac{m}{2^{i}},k,d) & =\sum_{i=1}^{\log\frac{m}{k}}m\cdot\frac{m}{2^{i}}\log\frac{1}{\eps}\cdot\left(O(\frac{1}{\eps})^{kd}+O(d\log\frac{m}{2^{i}})\right)\\
	& =m^{2}\log\frac{1}{\eps}\cdot\left(O(\frac{1}{\eps})^{kd}+O(d\log m)\right)~.
	\end{align*}
	
	\paragraph{Correctness.} We argue that $\dfd(P[i,j],Q)\le\tilde{\Delta}\le(1+\eps)\dfd(P[i,j],Q)$.
	
	\noindent
	If $\tilde{\Delta}=\min_{1\le q\le k} \max\{\mathcal{O}_1(Q[1,q]),\mathcal{O}_2(Q[q,k])\}$, then $$\tilde{\Delta}\ge\max\{\dfd(P[i,y],Q[1,q]),\dfd(P[y+1,j],Q[q,k])\}\ge\dfd(P[i,j],Q).$$Similarly, if $\tilde{\Delta}=\min_{1\le q\le k-1} \max\{\mathcal{O}_1(Q[1,q]),\mathcal{O}_2(Q[q+1,k])\}$ then 
	$$\tilde{\Delta}\ge \min_{1\le q\le k-1} \max\{\mathcal{O}_1(Q[1,q]),\mathcal{O}_2(Q[q+1,k])\}\ge \dfd(P[i,j],Q).$$
	Consider an optimal paired walk $\omega$ along $P[i,j]$ and $Q$, and let $1\le q\le k$ be the maximum index such that $\omega$ matches $P[y]$ and $Q[q]$.
	If $\omega$ matches $P[y+1]$ and $Q[q]$ then $$\dfd(P[i,j],Q)=\max\{\dfd(P[i,y],Q[1,q]),\dfd(P[y+1,j],Q[q,k])\}$$ and therefore $\max\{\mathcal{O}_1(Q[1,q]),\mathcal{O}_2(Q[q,k])\}\le (1+\eps)\dfd(P[i,j],Q)$.
	Else, it must be that $\omega$ matches $P[y+1]$ and $Q[q+1]$ (due to the maximality of $q$) and then $$\dfd(P[i,j],Q)=\max\{\dfd(P[i,y],Q[1,q]),\dfd(P[y+1,j],Q[q+1,k])\}$$ and therefore $\max\{\mathcal{O}_1(Q[1,q]),\mathcal{O}_2(Q[q+1,k])\}\le (1+\eps)\dfd(P[i,j],Q)$. 
	We conclude that $\tilde{\Delta}\le(1+\eps)\dfd(P[i,j],Q)$.
	
	\DOsubCurve*

%% file: High_D_Optimization.tex
	\section{High dimensional discrete \frechet\ algorithms}\label{sec:high_d_optimization}
	Most of the algorithms for curves under the (discrete) \frechet\ distance that were proposed in the literature, were only presented for constant or low dimension.
	The reason being is that it is often the case that the running time scale exponentially with the dimension (this  phenomena usually referred to as ``the curse of dimensionality'').
	
	In this section we provide a basic tool for finding a small set of critical values in $d$-dimensional space, and show how to apply it for tasks concerning approximation algorithms for curves under the discrete \frechet\ distance.
	While algorithms for those tasks exist in low dimensions, their generalization to high dimensional Euclidean space either do no exist or suffer from exponential dependence on the dimension.

	Chan and Rahmati~\cite{CR18} (improving Bringmann and Mulzer~\cite{BM16}) presented an algorithm that given two curves $P$ and $Q$ in $\reals^{d\times m}$, and a value $1\le f \le m$, finds a value $\tilde{\Delta}$ such that $\dfd(P,Q)\le\tilde{\Delta}\le f\dfd(P,Q)$, in time $O(m\log m+m^2/f^2)\cdot\exp(d)$. 
	Actually, their algorithm consist of two part: decision and optimization.
	Fortunately, the decision algorithm is only polynomial in $d$:
	\begin{theorem}[\cite{CR18}]\label{thm:approximation_desicion}
		Given two curves $P$ and $Q$ in $\reals^{d\times m}$, there exists an algorithm with running time $O(md+(md/f)^2d)$ that returns YES if $\dfd(P,Q)\le 1$, and NO if $\dfd(P,Q)\ge f$; If $1\le \dfd(P,Q)\le f$, the algorithm may return either YES or NO.
	\end{theorem}
	The optimization procedure is the one presented by Bringmann and Mulzer in \cite{BM16}, which adds an $O(m\log m)$ additive factor to the running time (for constant $d$). However, the running time of the optimization procedure depends exponentially on $d$.
	In \Cref{thm:approximation} we show that the exponential factor in the running time can be removed without affecting the approximation factor.
	\CrudeFreshetApproxHighDim*

	Bereg et. al.~\cite{BJWYZ08} presented an algorithm that computes in $O(mk\log m \log(m/k))$ time an optimal $k$-simplification of a curve $P\in\reals^{3\times m}$.
	In \Cref{thm:SimplificationHighDim} we improve the running time and generalize this result to arbitrary high dimension $d$, while allowing a $1+\eps$ approximation.
	Note that \cite{BJWYZ08} works only for dimension $d\le 3$, and for $k=\Omega(m)$ has quadratic running time, while our algorithm runs in essentially linear time (up to a polynomial dependence in $\eps$), for arbitrarily large dimension $d$.
	\SimplificationHighDim*

	The algorithms for both \Cref{thm:approximation} and \Cref{thm:SimplificationHighDim} use the following lemma.  
	\begin{lemma}\label{lem:allDistances}
		Consider a set $V$ of $n$ points in $\R^d$ and an interval $[a,b]\subset\R_+$. Then for every parameter $O\left(nd\cdot\left(\log \textsl{}n+\frac{1}{\eps}\log(\frac{b}{a}d)\right)\right)$ time algorithm, that returns a set $M\subset\R_+$ of $O(\frac{n}{\epsilon}d\log(d\frac{b}{a}))$ numbers such that for every pair of points $x,y\in V$ and a real number $\beta\in[a,b]$, there is a number $\alpha\in M$ such that $\alpha\le \beta\cdot\|x-y\|_2\le (1+\eps)\cdot \alpha$.
	\end{lemma}
	\begin{proof}
		For every $i\in [d]$, denote by $x_i$ the $i$'th coordinate of a point $x$, and let ${V_i=\{x_i\mid x\in V\}\subset\R}$. Set $\delta=\frac12$. We construct a $\frac1\delta$-WSPD (well separated pair decomposition) $\mathcal{W}_i$ for $V_i$. Specifically, $\mathcal{W}_i=\left\{\{A_1,B_1\},\dots,\{A_s,B_s\}\right\}$ is a set of $s\le\frac n\delta$ pairs of sets $A_j,B_j\subseteq V_i$ such that for every $x,y\in V_i$, there is a pair $\{A_j,B_j\}\in\mathcal{W}_i$ such that $x\in A_j$ and $y\in B_j$ (or vice versa), and for every $j\in[s]$, $\max\{\diam(A_j),\diam(B_j)\}\le\delta\cdot d(A_j,B_j)$, where $ d(A_j,B_j)=\min_{x\in A_j,y\in B_j}|x-y|$. Such a WSPD exists, and it can be constructed in $O(n\log n+\frac n\delta)$ time (see e.g. \cite{HarPeledBook11}, Theorem 3.10).

	Observe that by the definition of WSPD, and the triangle inequality, for any $\{A,B\}\in\mathcal{W}_i$ and two points $p\in A$ and $q\in B$, it holds that
	\begin{equation}\label{eq:wspd}
	d(A,B)\le\left|p-q\right|\le d(A,B)+\diam(A)+\diam(B)\le(1+2\delta)\cdot d(A,B)=2\cdot d(A,B)
	\end{equation}

	For each set $\mathcal{W}_i$ and pair $\{A,B\}\in \mathcal{W}_i$, pick some arbitrary points $x'\in A$ and $y'\in B$, and set $\delta_i=|x'-y'|$. By \Cref{eq:wspd}, 
	$d(A,B)\le \delta_i \le 2\cdot d(A,B)$.	
	
	Now for each $1\le i\le d$ set $$M_{i}=\left\{\delta_i \cdot(1+\eps)^{q}\mid\{A,B\}\in\mathcal{W}_{i},~q\in\left[\lfloor\log_{1+\eps}\frac a2\rfloor,\lfloor\log_{1+\eps}(2b\sqrt{d})\rfloor\right]\right\},$$
	and let $M=\bigcup_{i=1}^d M_i$.
	We argue that the set $M$ satisfies the condition of the lemma. Note that indeed $|M|\le d\cdot\frac{n}{\delta}\cdot\log_{1+\epsilon}(\frac{4b}{a}\sqrt{d})=O(\frac{n}{\epsilon}d\log(d\cdot\frac{b}{a}))$. Further, the construction time is 
	\[
	O(n\log n+\frac{n}{\delta})\cdot d+O(d\cdot\frac{n}{\delta}\cdot\log_{1+\epsilon}(\frac{4b}{a}\sqrt{d}))=O\left(nd\cdot\left(\log n+\frac{1}{\eps}\log(\frac{b}{a}d)\right)\right)~.
	\]
	
	Consider some pair $x,y\in V$ and let $i$ be the coordinate where $|x_i-y_i|$ is maximized. Then
	$
	\left|x_{i}-y_{i}\right|=\|x-y\|_{\infty}\le\|x-y\|_{2}\le\sqrt{d}\cdot\|x-y\|_{\infty}
	$.
	Let $\{A,B\}\in\mathcal{W}_i$ be a pair such that $x_{i}\in A$, and $y_{i}\in B$. By \Cref{eq:wspd}, 
	$d(A,B)\le \left|x_{i}-y_{i}\right| \le 2\cdot d(A,B)$, and therefore $\frac12 \delta_i\le \left|x_{i}-y_{i}\right| \le 2\delta_i$.	
	
	We conclude that $\frac12\delta_i\le\|x-y\|_{2}\le2\sqrt{d}\cdot \delta_i$. It follows that for every real parameter $\beta\in[a,b]$, there is a unique integer $q\in\left[\lfloor\log_{1+\eps}\frac{a}{2}\rfloor,\lfloor\log_{1+\eps}(2b\sqrt{d})\rfloor\right]$ such that		
	$$(1+\epsilon)^{q}\cdot \delta_i\le\beta\|x-y\|_{2}\le (1+\epsilon)^{q+1}\cdot \delta_i.$$ As $(1+\epsilon)^{q}\cdot \delta_i\in M$, the lemma follows.
	\end{proof}

	\subsection{Approximation algorithm: proof of \Cref{thm:approximation}}\label{subsec:m-approx}
	First, if $f\le 2$, we compute $\dfd(P,Q)$ exactly in $O(m^2d)$ time.
	Otherwise, we set $f'=f/2$.
	
	Next, we apply the algorithm from \Cref{lem:allDistances} for the points of $P\cup Q$ with parameter $\eps=\frac12$ and interval $[1,1]$, to obtain a set $M$ of $O(md\log d)$ scalars which is constructed in $O(md\log (md))$ time.
	Notice that there exists two points $x\in P$ and $y\in Q$ such that $\dfd(P,Q)=\|x-y\|$. Therefore, there exists $\alpha^*\in M$ such that $\alpha^*\le \dfd(P,Q) \le \frac32 \alpha^*<2\alpha^*$.
	
	Then, we sort the numbers in $M$ (in $O(md\log (md)\log d)$ time), and let $\alpha_1,\dots,\alpha_{|M|}$ be the sorted list of scalars. We call $\alpha_i$ a YES-entry if the algorithm from \Cref{thm:approximation_desicion} returns YES on the input $f'$ and $P,Q$ scaled by $2\alpha_i$, and NO-entry if it returns NO on this input.
	Notice that any $\alpha_i$ must be a NO-entry if $2\alpha_i\cdot f'\le\dfd(P,Q)$, and any $\alpha_i$ must be a YES-entry if $2\alpha_i\ge\dfd(P,Q)$. Moreover, $\alpha_{|M|}$ must be a YES-entry because $\dfd(P,Q) \le 2\alpha^*\le2\alpha_{|M|}$.
	
	If $\alpha_1$ is a YES-query, then $\dfd(P,Q)\le 2\alpha_1f'$.
	We return $\tilde{\Delta}=2\alpha_1f'$, and as $\alpha_1\le\alpha^*\le \dfd(P,Q)$  we get that $\dfd(P,Q)\le\tilde{\Delta}=2\alpha_1f'\le 2f'\dfd(P,Q)=f\dfd(P,Q)$.
	
	Else, using binary search on $\alpha_1,\dots,\alpha_{|M|}$, we find some $i$ such that $\alpha_i$ is a YES-entry and $\alpha_{i-1}$ is a NO-entry, and return $\tilde{\Delta}=2\alpha_if'$.
	Since $\alpha_i$ is a YES-entry, we have $\dfd(P,Q)\le 2\alpha_if'$.
	As $\alpha_{i-1}$ is a NO-entry, we have $\dfd(P,Q)> 2\alpha_{i-1}$, and thus $\alpha^*> \alpha_{i-1}$ and $\alpha^*\ge\alpha_i \ge \alpha_{i-1}$, so $\alpha_i\le\alpha^*\le \dfd(P,Q)$ and we get that $\dfd(P,Q)\le\tilde{\Delta}=2\alpha_if'\le 2f'\dfd(P,Q)=f\dfd(P,Q)$.
	This search takes $\log(md)\cdot O(md+(md/f)^{2}d)$ time.
	\qed

	\subsection{Computing a $(k,1+\eps)$-simplification: proof of \Cref{thm:SimplificationHighDim}}\label{sec:approx-simp}
	Bereg et al.~\cite{BJWYZ08} presented an algorithm that for constant $d$ computes an optimal $\delta$-simplification (this is a simple greedy simplification using Megiddo \cite{Meg84} linear time minimum enclosing ball algorithm). 
	The authors and Katz \cite{FFK20}, generalized this algorithm to high dimension $d$ by providing an algorithm, that given a scalar $\delta$, computes an approximation to the optimal $\delta$-simplification.
	\begin{lemma}[\cite{FFK20}]\label{lem:optrsimplification}
		Let $C$ be a curve consisting of $m$ points in $\R^d$. Given parameters $r>0$, and $\eps\in(0,1]$, there exists an algorithm that runs in $O\left(\frac{d\cdot m\log m}{\eps}+m\cdot\eps^{-4.5}\log\frac{1}{\eps}\right)$ time and returns a curve $\Pi$ such that $\dfd(C,\Pi)\le (1+\eps)r$. Furthermore, for every curve $\Pi'$ with $|\Pi'|<|\Pi|$, it holds that $\dfd(C,\Pi')>r$.
	\end{lemma}

	We begin with the following observation.
	\begin{observation}\label{obs:SimpInInterval}
		Consider a curve $P\in\R^{m\times d}$ and an optimal $k$-simplification $\Pi$ of $P$. Then there exists a pair of points $x,y\in P$ and a scalar $\beta\in[\frac12,\frac{1}{\sqrt{2}}]$ such that $\dfd(P,\Pi)=\beta\cdot\|x-y\|$.
	\end{observation}
	\begin{proof}
		Let $\omega$ be a one-to-many paired walk along $\Pi$ and $P$. If no such a walk exists, then we can simply remove vertices from $\Pi$ without increasing the distance to $P$.
		
		Notice that there must exist an index $1\le j\le k$ 
		such that $(\Pi[j],P[i_1,i_2])\in\omega$, and the minimum enclosing ball $B$ of $P[i_1,i_2]$ has radius $\dfd(P,\Pi)$. Otherwise, for each $j$ we can move $\Pi[j]$ to the center of the appropriate minimum enclosing ball $B$, and decrease the distance between $\Pi[j]$ and $P[i_1,i_2]$, which will decrease $\dfd(P,\Pi)$, in contradiction to the optimality of $\Pi$.
		
		Let $x,y\in P$ be the points such that $\|x-y\|=\max_{i_1\le p,q\le i_2}\|P[p]-P[q]\|$ (the diameter of $P[i_1,i_2]$). Then the radius of $B$ is at least $\frac{1}{2}\|x-y\|$, and by \href{https://mathworld.wolfram.com/JungsTheorem.html}{Jung's Theorem}, it is at most $\sqrt{\frac{d}{2(d+1)}}\cdot\|x-y\|<\sqrt{\frac{1}{2}}\cdot\|x-y\|$.
	\end{proof}

	\begin{proof}[Proof of \Cref{thm:SimplificationHighDim}]
		We first present the algorithm for general dimension $d$. Afterwards, we will reduce the running time for the case where $d$ is fixed.
		
		Fix $\eps'=\frac\eps3$. Using \Cref{lem:allDistances} for the points of $P$ with parameter $\eps'$ and interval $[\frac12,\frac{1}{\sqrt{2}}]$, we obtain a set $M$ of $O(\frac{m}{\epsilon}d\log d)$ scalars, which is constructed in $O\left(md\cdot\left(\log m+\frac{1}{\eps}\log d\right)\right)$
		time. 
		We sort the numbers in $M$, and using binary search
		we find the minimal $\alpha\in M$, such that the algorithm of \Cref{lem:optrsimplification} with parameter $(1+\eps')\alpha$ returns a simplification $\Pi_\alpha$ of length at most $k$ such that $\dfd(P,\Pi_\alpha)\le (1+\eps')^2\alpha$.
		We return the curve $\Pi_\alpha$ with parameter $(1+\eps')^2\alpha$.
		
		Let $\delta^*$ be the distance between $P$ and an optimal $k$-simplification of $P$. We argue that $\dfd(P,\Pi_\alpha)\le (1+\eps)\delta^*$.
		First note that by \Cref{obs:SimpInInterval} there exists a pair of points $x,y\in P$ and a scalar $\beta\in[\frac12,\frac{1}{\sqrt{2}}]$ such that $\delta^*=\beta\cdot\|x-y\|$. It follows from \Cref{lem:allDistances} that there is some $\alpha^*\in M$ such that $\alpha^*\le \delta^*\le(1+\eps')\cdot \alpha^*$. In particular, the algorithm from \Cref{lem:optrsimplification} with the parameter $(1+\eps')\alpha^*$ would return a curve $\Pi_{\alpha^*}$ of length at most $k$ such that $\dfd(P,\Pi_{\alpha^*})\le(1+\eps')^2\alpha^*$. Hence our algorithm will find some $\alpha\in M$, so that $\alpha\le\alpha^*$, and will return the curve $\Pi_\alpha$. It holds that 
		\[
		\dfd(P,\Pi_{\alpha})\le(1+\eps')^{2}\alpha\le(1+\eps')^{2}\alpha^{*}\le(1+\eps')^{2}\delta^{*}<(1+\eps)\delta^{*}~.
		\]
		
		Sorting $M$ takes $O(|M|\log|M|)=O(\frac{md}{\epsilon}\log\frac{md}{\epsilon}\cdot \log d)$ time.
		Finally, we have at most $\log |M|$ executions of the algorithm of \Cref{lem:optrsimplification} which take us $O(\log\frac{md}{\epsilon})\cdot O\left(\frac{d\cdot m\log m}{\eps}+m\cdot\eps^{-4.5}\log\frac{1}{\eps}\right)$
		time.
		The overall running time is
		\[
		O\left(\frac{md}{\epsilon}\log\frac{md}{\epsilon}\cdot\log md+m\log\frac{md}{\epsilon}\cdot\eps^{-4.5}\log\frac{1}{\eps}\right)=\tilde{O}\left(\frac{md}{\eps^{4.5}}\right)~.
		\]
		
		For the case where the dimension $d$ is fixed, instead of using \Cref{lem:optrsimplification}, we will simply find an optimal $1(\eps')\alpha$ simplification in $O(m\log m)$ time using \cite{BJWYZ08}. The reset of the algorithm will remain the same. The correctness proof follows the exact same lines. The running time will be $O(m\cdot(\log m+\frac{1}{\eps}))+\log(\frac{m}{\eps})\cdot O(m\log m)=m\cdot O(\frac{1}{\eps}+\log\frac{m}{\eps}\log m)$.
	\end{proof}

%% file: Missing_Proofs.tex
\section{Missing proofs}

\subsection{Proof of \Cref{clm:grid_points_in_ball}}\label{app:grid_points_in_ball}

\gridPointsInBall*
\begin{proof}
	We scale our grid so that the edge length is 1, hence we are looking for the number of lattice points in $B^d_2(x,\frac{c\sqrt{d}}{\eps})$. By Lemma 5 from \cite{FFK20} we get that this number is bounded by the volume of the $d$-dimensional ball of radius $\frac{c\sqrt{d}}{\eps}+\sqrt{d}\le \frac{(c+1)\sqrt{d}}{\eps}$.
	
	Using Stirling's formula we conclude that the volume of this ball is
	\[
	V^d_2 \left(\frac{(c+1)\sqrt{d}}{\eps} \right)=\frac{\pi^{d/2}}{\Gamma(\frac{d}{2}+1)}\left(\frac{(c+1)\sqrt{d}}{\eps} \right)^d = O\left(\frac{c}{\eps}\right)^d.
	\]
\end{proof}

\subsection{$(1+\eps)$-\MEB: Proof of \Cref{lem:1+epsMEB}}\label{app:MEBproof}

We begin with a definition of $\eps$-kernel of a set of points.
\begin{definition}[$\eps$-kernel]
	For a set of points $X\subseteq\R^d$, and a direction $\vec{u}\in\mathbb{S}^{d-1}$, the directional width of $X$ along $u$ is defined by $W(X,\vec{u}) =\max_{\vec{p},\vec{q}\in X}\langle \vec{p}-\vec{q},\vec{u}\rangle$ (here $\langle\cdot,\cdot\rangle$ denotes the inner product).
	A subset $Y\subseteq X$ of points is called an $\eps$-kernel of $X$ if for every direction $\vec{u}\in\mathbb{S}^{d-1}$, 
	$$W(Y,\vec{u})\ge(1-\eps)W(X,\vec{u})~.$$
\end{definition}
It was shown by \cite{Chan06}, that every set $X\subseteq\R^d$ has an $\eps$-kernel of size $O(\eps^{-\frac{d-1}{2}})$.
Zarrabi-Zadeh showed how to efficiently maintain an $\eps$-kernel in the streaming model.
\begin{theorem}[\cite{Zar11}]\label{thm:Zar11}
	Given a stream of points $X\subseteq\R^d$, an $\eps$-kernel of $X$ can be maintained using $O(\eps)^{-\frac{d-1}{2}}\cdot\log\frac1\eps$ space. 
	\footnote{Actually Zarrabi-Zadeh \cite{Zar11} bounds the space by $O(\eps^{-\frac{d-1}{2}})$, while assuming that $d$ is fixed, and hence hiding exponential factors in $d$.}
\end{theorem}

We make the following observation:
\begin{claim}\label{clm:KernalOfMEB}
	Consider a set $X\subseteq\R^d$, and let $Y$ be an $\eps$-kernel for $\eps\in(0,\frac14)$. Consider a ball $B(\vec{c},r)$ containing $Y$. Then $X\subseteq B(\vec{x},(1+3\eps)r)$.
\end{claim}
\begin{wrapfigure}{r}{0.28\textwidth}
	\begin{center}
		\vspace{-20pt}
		\includegraphics[width=0.27\textwidth]{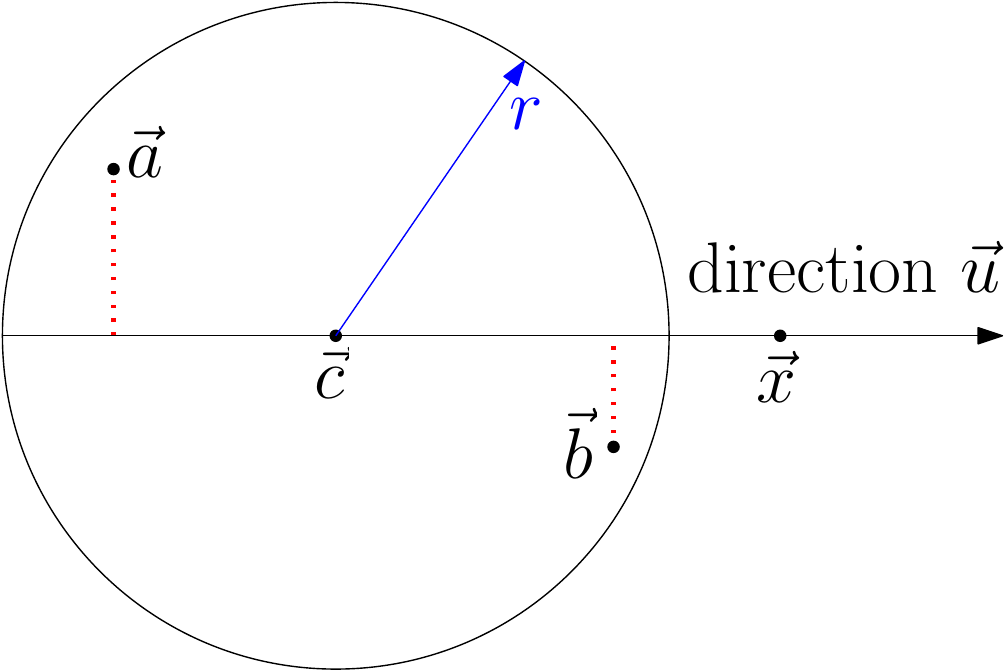}
		\vspace{-5pt}
	\end{center}
	\vspace{-10pt}
\end{wrapfigure}
\hspace{-18pt}\emph{Proof.}
	We will assume that $X$ is a finite set, the proof can be generalize to infinite sets using standard compactness arguments.
	Assume for contradiction that there is a point $\vec{x}\in X$ such that 
	$\vec{x}\notin B(x,(1+3\eps)r)$.
	Set $\vec{u}=\frac{\vec{x}-\vec{c}}{\|\vec{x}-\vec{c}\|}$, and let $\vec{a}={\argmin}_{\vec{a}\in Y}\langle \vec{a},\vec{u}\rangle$ and $\vec{b}={\argmax}_{\vec{p}\in Y}\langle \vec{p},\vec{u}\rangle$. See illustration on the right.
	Set $x=\langle \vec{x},\vec{u}\rangle$, $a=\langle \vec{a},\vec{u}\rangle$, and $b=\langle \vec{b},\vec{u}\rangle$. 
	As $\vec{x}\notin B(\vec{c},(1+3\eps)r)$, the distance between the  projection of $\vec{x}$ in direction $\vec{u}$ to the projection of every point in $B(\vec{c},(1+3\eps)r)$ is greater than $3\eps r$, thus $x-b>3\eps r$.
	From the other hand, as $\vec{a},\vec{b}$ contained in a ball of diameter $2r$, we have $W(Y,\vec{u})=b-a\le 2r$.
	Hence
	\[
	W(X,\vec{u})\ge x-a=(x-b)+(b-a)>3\eps r+W(Y,\vec{u})\ge(1+\frac{3\eps}{2})W(Y,\vec{u})~.
	\]
	In particular, $W(Y,\vec{u})<(1-\eps)W(X,\vec{u})$, 
\qed

\begin{proof}[Proof of \Cref{lem:1+epsMEB}]
	For a stream of points $X$, we will maintain an $\frac\eps5$-kernel $Y$ using \Cref{thm:Zar11}. On a query for a minimum enclosing ball, we will compute an enclosing  $B(\vec{c},(1+\frac\eps 5)r)$ for $Y$ (using \cite{KMY03}), such that there is no ball of radius $r$ enclosing $Y$ (or $X$). 
	By \Cref{clm:KernalOfMEB}, as $Y$ is an $\frac\eps5$ kernel of $X$, it holds that $X\subseteq B(\vec{c},(1+3\frac\eps5)(1+\frac\eps 5)r)\subseteq B(\vec{c},(1+\eps)r)$.
	We will return $B(\vec{c},(1+\eps)r)$ as an answer.
\end{proof}